\documentclass[a4paper,11pt]{article}
\usepackage[T1]{fontenc} 

\makeatletter
\gdef\@fpheader{ }
\gdef\@journal{ }
\makeatother
\RequirePackage{amsmath}
\RequirePackage{amssymb}
\RequirePackage{epsfig}
\RequirePackage{graphicx}
\RequirePackage[numbers,sort&compress]{natbib}
\RequirePackage{color}
\RequirePackage[colorlinks=true
,urlcolor=blue
,anchorcolor=blue
,citecolor=blue
,filecolor=blue
,linkcolor=blue
,menucolor=blue
,pagecolor=blue
,linktocpage=true
,pdfproducer=medialab
,pdfa=true
]{hyperref}

\newif\ifnotoc\notocfalse
\newif\ifemailadd\emailaddfalse
\newif\iftoccontinuous\toccontinuousfalse
\makeatletter
\def\@subheader{\@empty}
\def\@keywords{\@empty}
\def\@abstract{\@empty}
\def\@xtum{\@empty}
\def\@dedicated{\@empty}
\def\@arxivnumber{\@empty}
\def\@collaboration{\@empty}
\def\@collaborationImg{\@empty}
\def\@proceeding{\@empty}
\def\@preprint{\@empty}

\newcommand{\subheader}[1]{\gdef\@subheader{#1}}
\newcommand{\keywords}[1]{\if!\@keywords!\gdef\@keywords{#1}\else%
\PackageWarningNoLine{\jname}{Keywords already defined.\MessageBreak Ignoring last definition.}\fi}
\renewcommand{\abstract}[1]{\gdef\@abstract{#1}}
\newcommand{\dedicated}[1]{\gdef\@dedicated{#1}}
\newcommand{\arxivnumber}[1]{\gdef\@arxivnumber{#1}}
\newcommand{\proceeding}[1]{\gdef\@proceeding{#1}}
\newcommand{\xtumfont}[1]{\textsc{#1}}
\newcommand{\correctionref}[3]{\gdef\@xtum{\xtumfont{#1} \href{#2}{#3}}}
\newcommand\jname{JHEP}
\newcommand\acknowledgments{\section*{Acknowledgments}}

\newcommand\preprint[1]{\gdef\@preprint{\hfill #1}}

\newtheorem{theorem}{Theorem}

\newenvironment{proof}[1][Proof]{\noindent\textbf{#1.} }{\ \rule{0.5em}{0.5em}}

\makeatother

\newcommand\note[2][]{%
\if!#1!%
\stepcounter{footnote}\footnotetext{#2}%
\else%
{\renewcommand\thefootnote{#1}%
\footnotetext{#2}}%
\fi}

\makeatletter
\newtoks\auth@toks
\renewcommand{\author}[2][]{%
  \if!#1!%
    \auth@toks=\expandafter{\the\auth@toks#2\ }%
  \else
    \auth@toks=\expandafter{\the\auth@toks#2$^{#1}$\ }%
  \fi
}
\makeatother
\makeatletter
\newtoks\affil@toks\newif\ifaffil\affilfalse
\newcommand{\affiliation}[2][]{%
\affiltrue
  \if!#1!%
    \affil@toks=\expandafter{\the\affil@toks{\item[]#2}}%
  \else
    \affil@toks=\expandafter{\the\affil@toks{\item[$^{#1}$]#2}}%
  \fi
}
\makeatother
\makeatletter
\newtoks\email@toks\newcounter{email@counter}%
\newcommand{\emailAdd}[1]{%
\emailaddtrue%
\ifnum\theemail@counter>0\email@toks=\expandafter{\the\email@toks, \@email{#1}}%
\else\email@toks=\expandafter{\the\email@toks\@email{#1}}%
\fi\stepcounter{email@counter}}
\newcommand{\@email}[1]{\href{mailto:#1}{\tt #1}}
\makeatother

\makeatletter
\newcommand*\collaboration[1]{\gdef\@collaboration{#1}}
\newcommand*\collaborationImg[2][]{\gdef\@collaborationImg{#2}}
\makeatletter
\newcommand\afterLogoSpace{\smallskip}
\newcommand\afterSubheaderSpace{\vskip3pt plus 2pt minus 1pt}
\newcommand\afterProceedingsSpace{\vskip21pt plus0.4fil minus15pt}
\newcommand\afterTitleSpace{\vskip23pt plus0.06fil minus13pt}
\newcommand\afterRuleSpace{\vskip23pt plus0.06fil minus13pt}
\newcommand\afterCollaborationSpace{\vskip3pt plus 2pt minus 1pt}
\newcommand\afterCollaborationImgSpace{\vskip3pt plus 2pt minus 1pt}
\newcommand\afterAuthorSpace{\vskip5pt plus4pt minus4pt}
\newcommand\afterAffiliationSpace{\vskip3pt plus3pt}
\newcommand\afterEmailSpace{\vskip16pt plus9pt minus10pt\filbreak}
\newcommand\afterXtumSpace{\par\bigskip}
\newcommand\afterAbstractSpace{\vskip16pt plus9pt minus13pt}
\newcommand\afterKeywordsSpace{\vskip16pt plus9pt minus13pt}
\newcommand\afterArxivSpace{\vskip3pt plus0.01fil minus10pt}
\newcommand\afterDedicatedSpace{\vskip0pt plus0.01fil}
\newcommand\afterTocSpace{\bigskip\medskip}
\newcommand\afterTocRuleSpace{\bigskip\bigskip}
\newlength{\affiliationsSep}
\newcommand\beforetochook{\pagestyle{myplain}\pagenumbering{roman}}

\DeclareFixedFont\trfont{OT1}{phv}{b}{sc}{11}

\renewcommand\maketitle{
\pagestyle{empty}
\thispagestyle{titlepage}
\setcounter{page}{0}
\noindent{\small\scshape\@fpheader}\@preprint\par

\afterLogoSpace
\if!\@subheader!\else\noindent{\trfont{\@subheader}}\fi
\afterSubheaderSpace
\if!\@proceeding!\else\noindent{\sc\@proceeding}\fi
\afterProceedingsSpace
{\LARGE\flushleft\sffamily\bfseries\@title\par}
\afterTitleSpace
\hrule height 1.5\p@%
\afterRuleSpace
\if!\@collaboration!\else
{\Large\bfseries\sffamily\raggedright\@collaboration}\par
\afterCollaborationSpace
\fi
\if!\@collaborationImg!\else
{\normalsize\bfseries\sffamily\raggedright\@collaborationImg}\par
\afterCollaborationImgSpace
\fi
{\bfseries\raggedright\sffamily\the\auth@toks\par}
\afterAuthorSpace
\ifaffil\begin{list}{}{%
\setlength{\leftmargin}{0.28cm}%
\setlength{\labelsep}{0pt}%
\setlength{\itemsep}{\affiliationsSep}%
\setlength{\topsep}{-\parskip}}
\itshape\small%
\the\affil@toks
\end{list}\fi
\afterAffiliationSpace
\ifemailadd 
\noindent\hspace{0.28cm}\begin{minipage}[l]{.9\textwidth}
\begin{flushleft}
\textit{E-mail:} \the\email@toks
\end{flushleft}
\end{minipage}
\else 
\PackageWarningNoLine{\jname}{E-mails are missing.\MessageBreak Plese use \protect\emailAdd\space macro to provide e-mails.}
\fi
\afterEmailSpace
\if!\@xtum!\else\noindent{\@xtum}\afterXtumSpace\fi
\if!\@abstract!\else\noindent{\renewcommand\baselinestretch{.9}\textsc{Abstract:}}\ \@abstract\afterAbstractSpace\fi
\if!\@keywords!\else\noindent{\textsc{Keywords:}} \@keywords\afterKeywordsSpace\fi
\if!\@arxivnumber!\else\noindent{\textsc{ArXiv ePrint:}} \href{http://arxiv.org/abs/\@arxivnumber}{\@arxivnumber}\afterArxivSpace\fi
\if!\@dedicated!\else\vbox{\small\it\raggedleft\@dedicated}\afterDedicatedSpace\fi
\ifnotoc\else
\iftoccontinuous\else\newpage\fi
\beforetochook\hrule
\tableofcontents
\afterTocSpace
\hrule
\afterTocRuleSpace
\fi
\setcounter{footnote}{0}
\pagestyle{myplain}\pagenumbering{arabic}
} 

\renewcommand{\baselinestretch}{1.1}\normalsize
\divide\@tempdima\baselineskip
\@tempcnta=\@tempdima
\skip\@mpfootins = \skip\footins
\renewcommand{\@dotsep}{10000}

\newcommand\ps@myplain{
\pagenumbering{arabic}
\renewcommand\@oddfoot{\hfill-- \thepage\ --\hfill}
\renewcommand\@oddhead{}}
\let\ps@plain=\ps@myplain

\newcommand\ps@titlepage{\renewcommand\@oddfoot{}\renewcommand\@oddhead{}}

\numberwithin{equation}{section}

\renewcommand\section{\@startsection{section}{1}{\z@}%
                                   {-3.5ex \@plus -1.3ex \@minus -.7ex}%
                                   {2.3ex \@plus.4ex \@minus .4ex}%
                                   {\normalfont\large\bfseries}}
\renewcommand\subsection{\@startsection{subsection}{2}{\z@}%
                                   {-2.3ex\@plus -1ex \@minus -.5ex}%
                                   {1.2ex \@plus .3ex \@minus .3ex}%
                                   {\normalfont\normalsize\bfseries}}
\renewcommand\subsubsection{\@startsection{subsubsection}{3}{\z@}%
                                   {-2.3ex\@plus -1ex \@minus -.5ex}%
                                   {1ex \@plus .2ex \@minus .2ex}%
                                   {\normalfont\normalsize\bfseries}}
\renewcommand\paragraph{\@startsection{paragraph}{4}{\z@}%
                                   {1.75ex \@plus1ex \@minus.2ex}%
                                   {-1em}%
                                   {\normalfont\normalsize\bfseries}}
\renewcommand\subparagraph{\@startsection{subparagraph}{5}{\parindent}%
                                   {1.75ex \@plus1ex \@minus .2ex}%
                                   {-1em}%
                                   {\normalfont\normalsize\bfseries}}

\def\fnum@figure{\textbf{\figurename\nobreakspace\thefigure}}
\def\fnum@table{\textbf{\tablename\nobreakspace\thetable}}

\long\def\@makecaption#1#2{%
  \vskip\abovecaptionskip
  \sbox\@tempboxa{\small #1. #2}%
  \ifdim \wd\@tempboxa >\hsize
    \small #1. #2\par
  \else
    \global \@minipagefalse
    \hb@xt@\hsize{\hfil\box\@tempboxa\hfil}%
  \fi
  \vskip\belowcaptionskip}

\renewenvironment{thebibliography}[1]{%
\begin{oldthebibliography}{#1}%
\small%
\raggedright%
\setlength{\itemsep}{5pt plus 0.2ex minus 0.05ex}%
}%
{%
\end{oldthebibliography}%
}

\begin{document}



\renewcommand{\thefootnote}{\fnsymbol{footnote}}

\title{\boldmath Canonical partition functions: ideal quantum gases, interacting classical
gases, and interacting quantum gases}%

\author[]{Chi-Chun Zhou}
\author[*]{and Wu-Sheng Dai}\note{daiwusheng@tju.edu.cn.}


\affiliation[]{Department of Physics, Tianjin University, Tianjin 300350, P.R. China}









\abstract{In statistical mechanics, for a system with fixed number of particles, e.g., a
finite-size system, strictly speaking, the thermodynamic quantity needs to be
calculated in the canonical ensemble. Nevertheless, the calculation of the
canonical partition function is difficult.\textbf{ }In this paper, based on
the mathematical theory of the symmetric function, we suggest a method for the
calculation of the canonical partition function of ideal quantum gases,
including ideal Bose, Fermi, and Gentile gases. Moreover, we express the
canonical partition functions of interacting classical and quantum gases given
by the classical and quantum cluster expansion methods in terms of\ the Bell
polynomial in mathematics. The virial coefficients of ideal Bose, Fermi, and
Gentile gases is calculated from the exact canonical partition function. The
virial coefficients of interacting classical and quantum gases is calculated
from the canonical partition function by using the expansion of the Bell
polynomial, rather than calculated from the grand canonical potential.
}

\maketitle
\flushbottom


\section{Introduction}

It is often difficult to calculate the canonical partition function. In the
canonical ensemble, there is a constraint on the total number of particles.
Once there also exist interactions, whether quantum exchange interactions or
classical inter-particle interactions, the calculation in canonical ensembles
becomes complicated. In ideal quantum gases, there are quantum exchange
interactions among gas molecules; in interacting classical gases, there are
classical interactions among gas molecules; in interacting quantum gases,
there are both quantum exchange interactions and classical inter-particle
interactions among gas molecules. In an interacting system there is no
single--particle state, and in a canonical ensemble there is a constraint on
the total particle number. Therefore, in the calculation of the canonical
partition function for such systems, we have to take these two factors into
account simultaneously, and this makes the calculation difficult. It is a
common practice to avoid such a difficulty by, instead of canonical ensembles,
turning to grand canonical ensembles in which there is no constraint on the
total particle number. In grand canonical ensembles, instead of the particle
number, one uses the average particle number. In this paper, we suggest a
method to calculate canonical partition functions.

The method suggested in the present paper is based on the mathematical theory
of the symmetric function and the Bell polynomial.

The symmetric function was first studied by P. Hall in the 1950s
\cite{Littlewood2005The,Macdonald1998Symmetric}. Now it is studied as the ring
of symmetric functions in algebraic combinatorics. The symmetric function is
closely related to the integer partition in number theory and plays an
important role in the theory of group representations \cite{Andrews1976The}.

The Bell polynomial is a special function in combinatorial mathematics
\cite{wheeler1987bell,Collins2001The}. It is useful in the study of set
partitions and is related to the Stirling and Bell number. For more details of
the Bell polynomial, one can refer to Refs.
\cite{wheeler1987bell,Collins2001The,voinov1994power,abbas2005new,kruchinin2011derivation}%
.

First, using the method, we will calculate the exact canonical partition
function for ideal quantum gases, including ideal Bose gases, ideal Fermi
gases, and ideal Gentile gases. In ideal quantum gases, there exist quantum
exchange interactions, as will be shown later, the canonical partition
functions are symmetric functions. Based on the theory of the symmetric
function, the canonical partition function is represented as\ a linear
combination of $S$-functions. The $S$-function is an important class of
symmetric functions, and they are closely related to integer partitions
\cite{Andrews1976The} and permutation groups
\cite{Littlewood2005The,Macdonald1998Symmetric}.

Second, we show that the canonical partition functions of interacting
classical gases and interacting quantum gases given by the classical and
quantum cluster expansions are the Bell polynomial. In interacting classical
gases, there exist classical inter-particle interactions and in interacting
quantum gases, there exist both quantum exchange interactions and classical
inter-particle interactions. After showing that the canonical partition
function of interacting gases is the Bell polynomial, we can calculate, e.g.,
the virial coefficient by the property of the Bell polynomial.

We will calculate the virial coefficients of ideal Bose gases, ideal Fermi
gases, and ideal Gentile gases from their exact canonical partition functions.
Moreover, based on the property of the Bell polynomial, we directly calculate
the virial coefficients of interacting classical and quantum gases from the
canonical partition function directly, rather than, as in the cluster
expansion method, from the grand canonical potential. We also compare the
virial coefficients calculated in the canonical ensemble with those calculated
in the grand canonical ensemble. From the results one can see that the virial
coefficients are different at small $N$ and they are in consistency with each
other as $N$ goes to infinity.

There are many studies on canonical partition functions. Some canonical
partition functions of certain statistical models are calculated, for
examples, the canonical partition function for a two-dimensional binary
lattice \cite{kaufman1949crystal}, the canonical partition function for
quon\ statistics \cite{goodison1994canonical}, a general formula for the
canonical partition function expressed as sums of the $S$-function for a
parastatistical system \cite{chaturvedi1996canonical}, the canonical partition
function of freely jointed chain model \cite{mazars1998canonical}, an exact
canonical partition function for ideal Bose gases calculated by the recursive
method \cite{park2010thermodynamic}, and the canonical partition function of
fluids calculated by simulations \cite{do2011rapid}. Some methods for the
calculation of the\textbf{ }canonical partition function are developed, for
examples, the numerical method
\cite{park2010thermodynamic,glaum2007condensation}, the recursion relation of
the canonical partition function
\cite{glaum2007condensation,pratt2000canonical,wang2009thermodynamics}. The
behavior of the canonical partition function is also discussed, for example,
the zeroes of the canonical partition function
\cite{lee2012exact,taylor2013partition} and the classical limit of the
quantum-mechanical canonical partition function \cite{seglar2013classical}.

There are also many studies on the symmetric function and the Bell polynomial.
In mathematics, some studies devote to the application of the $S$-function,
also known as the Schur function, for example, the application of $S$-function
in the symmetric function space \cite{lapointe2003schur}, the application of
the $S$-function in probability and statistics \cite{proschan1977schur}, and
the supersymmetric $S$-function \cite{moens2007supersymmetric}. In physics,
there are also applications of the $S$-function, for example, the canonical
partition function of a parastatistical system is expressed as a sum of
$S$-functions \cite{chaturvedi1996canonical}, the factorial $S$-function, a
generalizations of the $S$-function, times a deformation of the Weyl
denominator may be expressed as the partition function of a particular
statistical-mechanical system \cite{bump2011factorial} and the application to
statistical mechanics and representation theory \cite{gorinasymptotics}. The
Bell polynomial can be used in solving the water wave equation
\cite{yun2013integrability}, in seventh-order Sawada--Kotera--Ito equation
\cite{shen2014bell}, in combinatorial Hopf algebras \cite{aboud2014bell}.
Moreover, various other applications can be found in
\cite{connon2010various,qi2015several,vasudevan1984combinants,collins2001role}.

There are studies devoted to the application of the Bell polynomial in
statistical mechanics. The canonical partition function of an interacting
system is expressed in terms of the Bell polynomial by expanding the grand
canonical partition function in power series of the fugacity $z$
\cite{fronczak2013cluster}. As further applications of the bell polynomial, a
microscopic interpretation of the grand potential is given in Ref.
\cite{fronczak2012microscopic}. The partition function of the zero-field Ising
model on the infinite square lattice is given by resorting to the Bell
polynomial \cite{siudem2016exact}. Ref. \cite{siudem2013partition} proves an
assertion suggested in Ref. \cite{fronczak2012microscopic}, which provides a
formula for the canonical partition function using Bell polynomials and shows
that the Bell polynomial is very useful in statistical mechanics. In this
paper, we will compare our method in obtaining the canonical partition
function of interacting gases with that in these literature.

This paper is organized as follow. In Sec. \ref{review}, we give a brief
review for the integer partition and the symmetric function. In Secs.
\ref{cBg} and \ref{cFg}, the canonical partition functions of ideal Bose and
Fermi gases are given. In Sec. \ref{cGg}, the canonical partition function of
ideal Gentile gases is given. In Secs. \ref{classical} and \ref{interacting},
we show that the canonical partition functions of interacting classical and
quantum gases are indeed Bell polynomials. In Sec. \ref{virial}, we calculate
the virial coefficients of ideal Bose, Fermi, and Gentile gases in the
canonical ensemble and compare them with those calculated in the grand
canonical ensemble. In Sec. \ref{virialint}, we calculate the virial
coefficients of interacting classical and quantum gases in the canonical
ensemble and compare them with those calculated in the grand canonical
ensemble. The conclusions are summarized in Sec. \ref{Conclusionoutlook}. In
appendix \ref{appendix}, the details of the calculation about Gentile gases is given.

\section{The integer partition and the symmetric function: a brief review
\label{review}}

The main result of the present paper is to represent the canonical partition
function as a linear combination of symmetric functions. The symmetric
function is closely related to the problem of integer partitions. In this
section, we give a brief review of symmetric functions and integer partitions.
For more detail of this part, one can refer to Refs.
\cite{Andrews1976The,Littlewood2005The,Macdonald1998Symmetric}.

\textit{Integer partitions.} The integer partitions of $N$ are representations
of $N$ in terms of other positive integers which sum up to $N$. The integer
partition is denoted by $\left(  \lambda\right)  =\left(  \lambda_{1}%
,\lambda_{2},...\right)  $, where $\lambda_{1}$, $\lambda_{2}$, $...$ are
called elements, and they are arranged in descending order. For example, the
integer partition of $4$ are $\left(  \lambda\right)  =\left(  4\right)  $,
$\left(  \lambda\right)  ^{\prime}=\left(  3,1\right)  $, $\left(
\lambda\right)  ^{\prime\prime}=\left(  2^{2}\right)  $, $\left(
\lambda\right)  ^{\prime\prime\prime}=\left(  2,1^{2}\right)  $, and $\left(
\lambda\right)  ^{\prime\prime\prime\prime}=\left(  1^{4}\right)  $, where,
e.g., the superscript in $1^{2}$ means $1$ appearing twice, the superscript in
$2^{2}$ means $2$ appearing twice, and so on.

\textit{Arranging integer partitions in a prescribed order.} An integer $N$
has many integer partitions and the unrestricted partition function $P\left(
N\right)  $ counts the number of integer partitions \cite{Andrews1976The}. For
a given $N$, one\ arranges the integer partition in the following order:
$\left(  \lambda\right)  $, $\left(  \lambda\right)  ^{\prime}$, when
$\lambda_{1}>\lambda_{1}^{\prime}$; $\left(  \lambda\right)  $, $\left(
\lambda\right)  ^{\prime}$, when $\lambda_{1}=\lambda_{1}^{\prime}$ but
$\lambda_{2}>\lambda_{2}^{\prime}$; and so on. One keeps comparing
$\lambda_{i}$ and $\lambda_{i}^{\prime}$ until all the integer partitions of
$N$ are arranged in a prescribed order. We denote $\left(  \lambda\right)
_{J}$ the $Jth$ integer partition of $N$. For example, the integer partitions
arranged in a prescribed order of $4$ are $\left(  \lambda\right)
_{1}=\left(  4\right)  $, $\left(  \lambda\right)  _{2}=\left(  3,1\right)  $,
$\left(  \lambda\right)  _{3}=\left(  2^{2}\right)  $, $\left(  \lambda
\right)  _{4}=\left(  2,1^{2}\right)  $, and $\left(  \lambda\right)
_{5}=\left(  1^{4}\right)  $. For any integer $N$, the first integer partition
is always $\left(  \lambda\right)  _{J=1}=\left(  N\right)  $ and the last
integer partition is always $\left(  \lambda\right)  _{J=P\left(  N\right)
}=\left(  1^{N}\right)  $. We denote $\lambda_{J,i}$ the $ith$ element in the
integer partition $\left(  \lambda\right)  _{J}$. For example, for the second
integer partition of $4$, i.e., $\left(  \lambda\right)  _{2}=\left(
3,1\right)  $, the elements of $\left(  \lambda\right)  _{2}$ are
$\lambda_{2,1}=3$ and $\lambda_{2,2}=1$, respectively. For any integer $N$, we
always have $\lambda_{1,1}=N$, and $\lambda_{P\left(  N\right)  ,1}%
=\lambda_{P\left(  N\right)  ,2}=...=\lambda_{P\left(  N\right)  ,N}=1$.

\textit{The symmetric function.} A function $f\left(  x_{1},x_{2}%
,...,x_{n}\right)  $ is called symmetric functions if it is invariant under
the action of the permutation group $S_{n}$; that is, for $\sigma\in S_{n}$,
\begin{equation}
\sigma f\left(  x_{1},x_{2},..,x_{n}\right)  =f\left(  x_{\sigma\left(
1\right)  },x_{\sigma\left(  2\right)  },..,x_{\sigma\left(  n\right)
}\right)  =f\left(  x_{1},x_{2},..,x_{n}\right)  ,
\end{equation}
where $x_{1}$, $x_{2}$,..., $x_{n}$ are $n$ independent variables of $f\left(
x_{1},x_{2},...,x_{n}\right)  $.

\textit{The symmetric function }$m_{\left(  \lambda\right)  }\left(
x_{1},x_{2},...,x_{l}\right)  $\textit{.} The symmetric function $m_{\left(
\lambda\right)  }\left(  x_{1},x_{2},...,x_{l}\right)  $ is an important kind
of symmetric functions. One of its definition is \cite{Macdonald1998Symmetric}%
\begin{equation}
m_{\left(  \lambda\right)  _{I}}\left(  x_{1},x_{2},...,x_{l}\right)
=\sum_{perm}x_{i_{1}}^{\lambda_{I1}}x_{i_{2}}^{\lambda_{I2}}....x_{i_{N}%
}^{\lambda_{IN}}, \label{mdc}%
\end{equation}
where $\sum_{perm}$ indicates that the summation runs over all possible
monotonically increasing permutations of $x_{i}$, $\left(  \lambda\right)
_{I}$ is the $Ith$ integer partition of $N$, and $\lambda_{Ij}$ is the $jth$
element in $\left(  \lambda\right)  _{I}$. Here, the number of variables is
$l$, and $l$ should be larger than $N$ as required in the definition. $l$, the
number of $x_{i}$, could be infinite and in the problem considered in this
paper, $l$ is always infinite. Each integer partition $\left(  \lambda\right)
_{I}$ of $N$ corresponds to a symmetric function $m_{\left(  \lambda\right)
_{I}}\left(  x_{1},x_{2},...\right)  $ and vise versa.

\textit{The }$S$\textit{-function }$\left(  \lambda\right)  \left(
x_{1},x_{2},...\right)  $\textit{. }The $S$-function $\left(  \lambda\right)
\left(  x_{1},x_{2},...\right)  $, another important kind of symmetric
functions, among their many definitions, can be defined as
\cite{Littlewood2005The,Macdonald1998Symmetric}
\begin{equation}
\left(  \lambda\right)  _{I}\left(  x_{1},x_{2},...\right)  =\sum
_{J=1}^{P\left(  N\right)  }\frac{g_{I}}{N!}\chi_{J}^{I}%
{\displaystyle\prod\limits_{m=1}^{k}}
\left(  \sum_{i}x_{i}^{m}\right)  ^{a_{J,m}}, \label{ss}%
\end{equation}
where $\left(  \lambda\right)  _{I}$ is the $Ith$ integer partition of $N$,
$a_{J,m}$ counts the times of the number $m$ appeared in $\left(
\lambda\right)  _{J}$, and $\chi_{J}^{I}$ is the simple characteristic of the
permutation group of order $N$. $g_{I}$ is defined as $g_{I}=N!\left(
{\displaystyle\prod\limits_{j=1}^{N}}
j^{a_{I,j}}a_{I,j}!\right)  ^{-1}$. Here we only consider the case that the
number of variables is always infinite. Each integer partition $\left(
\lambda\right)  _{I}$ of $N$ corresponds to an $S$-function $\left(
\lambda\right)  _{I}\left(  x_{1},x_{2},...\right)  $ and vice versa.

\textit{The relation between the }$S$\textit{-function }$\left(
\lambda\right)  \left(  x_{1},x_{2},...\right)  $\textit{ and the symmetric
function }$m_{\left(  \lambda\right)  }\left(  x_{1},x_{2},...\right)  $.
There is a relation between the $S$-function $\left(  \lambda\right)  \left(
x_{1},x_{2},...\right)  $ and the symmetric function $m_{\left(
\lambda\right)  }\left(  x_{1},x_{2},...\right)  $: the $S$-function $\left(
\lambda\right)  \left(  x_{1},x_{2},...\right)  $ can be expressed as a linear
combination of the symmetric function\textit{ }$m_{\left(  \lambda\right)
}\left(  x_{1},x_{2},...\right)  $,%

\begin{equation}
\left(  \lambda\right)  _{K}\left(  x_{1},x_{2},...\right)  =\sum
_{I=1}^{P\left(  N\right)  }k_{K}^{I}m_{\left(  \lambda\right)  _{I}}\left(
x_{1},x_{2},...\right)  , \label{smgx}%
\end{equation}
where the coefficient $k_{K}^{I}$ is the Kostka number
\cite{Macdonald1998Symmetric}.

\section{The canonical partition function of ideal Bose gases \label{cBg}}

In this section, we present an exact expression of the canonical partition
function of ideal Bose gases.

\begin{theorem}
For ideal Bose gases, the canonical partition function is%
\begin{equation}
Z_{BE}\left(  \beta,N\right)  =\left(  N\right)  \left(  e^{-\beta
\varepsilon_{1}},e^{-\beta\varepsilon_{2}},...\right)  , \label{1111}%
\end{equation}
where $\left(  N\right)  \left(  x_{1},x_{2},...\right)  $ is the $S$-function
corresponding to the integer partition $\left(  N\right)  $ defined by Eq.
(\ref{ss}) and $\varepsilon_{i}$ is the single-particle eigenvalue.
\end{theorem}

\begin{proof}
By definition, the canonical partition function of an $N$-body system is%
\begin{equation}
Z\left(  \beta,N\right)  =\sum_{s}e^{-\beta E_{s}}, \label{zzdy}%
\end{equation}
where $E_{s}$ denotes the eigenvalue of the system with the subscript $s$
labeling the states, $\beta=1/\left(  kT\right)  $, $k$ is the Boltzmann
constant, and $T$ is the temperature.

Rewrite Eq. (\ref{zzdy}) in terms of the occupation number. For an ideal
quantum gas, particles are randomly distributed on the single-particle states.
Collect the single-particle states occupied by particles by weakly increasing
order of energy, i.e., $\varepsilon_{i_{1}}\leq\varepsilon_{i_{2}}%
\leq\varepsilon_{i_{3}}\leq\cdots$, and denote the number of particle
occupying the $jth$ state by $a_{j}$. For Bose gases, there is no restriction
on $a_{j}$, since the maximum occupation number is infinite. Note that here
$a_{i}\geq1$ rather than $a_{i}\geq0$. This is because here only the occupied
states are reckoned in. Eq. (\ref{zzdy}) can be re-expressed as
\cite{schrodinger1989statistical}%
\begin{equation}
Z_{BE}\left(  \beta,N\right)  =\sum_{s}e^{-\beta E_{s}}=\sum_{\left\{
a_{i}\right\}  }\sum_{perm}e^{-\beta\varepsilon_{i_{1}}a_{1}-\beta
\varepsilon_{i_{2}}a_{2}.....}, \label{bszzdy}%
\end{equation}
where the sum $\sum_{perm}$ runs over all possible mononical increasing
permutation of $\varepsilon_{i}$ and the sum $\sum_{\left\{  a_{i}\right\}  }$
runs over all the possible occupation numbers restricted by the constraint%
\begin{equation}
\sum a_{i}=N,\text{ \ \ }a_{i}\geq1. \label{constraint}%
\end{equation}
The constraint (\ref{constraint}) ensures that the total number of particles
in the system is a constant.

Equalling the occupation number $a_{i}$ and the element $\lambda_{i}$,
equalling the variable $x_{i}$ and $e^{-\beta\varepsilon_{i}}$, and comparing
the canonical partition function, Eq. (\ref{bszzdy}), with the definition of
the symmetric function $m_{\left(  \lambda\right)  }\left(  x_{1}%
,x_{2},...\right)  $, Eq. (\ref{mdc}), give
\begin{equation}
Z_{BE}\left(  \beta,N\right)  =\sum_{I=1}^{P\left(  N\right)  }m_{\left(
\lambda\right)  _{I}}\left(  e^{-\beta\varepsilon_{1}},e^{-\beta
\varepsilon_{2}},...\right)  . \label{bs mid}%
\end{equation}
Substituting Eq. (\ref{smgx}) into Eq. (\ref{bs mid}) and setting all the
coefficients $k_{K}^{I}$ to $1$ give Eq. (\ref{1111}), where the relation
$\sum_{I=1}^{P\left(  N\right)  }m_{\left(  \lambda\right)  _{I}}\left(
e^{-\beta\varepsilon_{1}},e^{-\beta\varepsilon_{2}},...\right)  =\left(
\lambda\right)  _{1}\left(  e^{-\beta\varepsilon_{1}},e^{-\beta\varepsilon
_{2}},...\right)  $ \cite{Macdonald1998Symmetric} and $\left(  \lambda\right)
_{1}=\left(  N\right)  $ are used.
\end{proof}

In a word, the canonical partition function of ideal Bose gases is a $S$-function.

Moreover, by one kind of expression of the $S$-function given in Ref.
\cite{Littlewood2005The,Macdonald1998Symmetric}, we can obtain the canonical
partition function of ideal Bose gases given in Ref.
\cite{park2010thermodynamic}. In Ref. \cite{park2010thermodynamic}, the author
obtained a recurrence relation of the canonical partition function:
$Z_{BE}\left(  \beta,N\right)  =\frac{1}{N}\sum_{k=1}^{N}\left(  \sum
_{n}e^{-\beta k\varepsilon_{n}}\right)  Z_{BE}\left(  \beta,N-k\right)  $,
based on a result given by Mastsubara \cite{matsubara1951quantum} and Feynman
\cite{feynman2005statistical}. Starting from this recurrence relation, the
author obtains the canonical partition function of ideal Bose gases in matrix
form by introducing an $\infty\times\infty$ triangularized matrix and
computing the inverse of a matrix.

In Refs. \cite{Littlewood2005The,Macdonald1998Symmetric}, the $S$-function is
represented as the determinant of a certain matrix:
\begin{equation}
\left(  N\right)  \left(  x_{1},x_{2},...\right)  =\frac{1}{N!}\det\left(
\begin{array}
[c]{ccccc}%
\sum_{i}x_{i} & -1 & 0 & ... & 0\\
\sum_{i}x_{i}^{2} & \sum_{i}x_{i} & -2 & ... & 0\\
\sum_{i}x_{i}^{3} & \sum_{i}x_{i}^{2} & \sum_{i}x_{i} & ... & ...\\
... & ... & ... & ... & -\left(  N-1\right) \\
\sum_{i}x_{i}^{N} & \sum_{i}x_{i}^{N-1} & \sum_{i}x_{i}^{N-2} & ... & \sum
_{i}x_{i}%
\end{array}
\right)  . \label{Shanshu}%
\end{equation}
The canonical partition function (\ref{1111}), by Eq. (\ref{Shanshu}), can be
equivalently expressed as
\begin{align}
Z_{BE}\left(  \beta,N\right)   &  =\left(  N\right)  \left(  e^{-\beta
\varepsilon_{1}},e^{-\beta\varepsilon_{2}},...\right) \nonumber\\
&  =\frac{1}{N!}\det\left(
\begin{array}
[c]{ccccc}%
Z\left(  \beta\right)  & -1 & 0 & ... & 0\\
Z\left(  2\beta\right)  & Z\left(  \beta\right)  & -2 & ... & 0\\
Z\left(  3\beta\right)  & Z\left(  2\beta\right)  & Z\left(  \beta\right)  &
... & ...\\
... & ... & ... & ... & -\left(  N-1\right) \\
Z\left(  N\beta\right)  & Z\left(  N\beta-\beta\right)  & Z\left(
N\beta-2\beta\right)  & ... & Z\left(  \beta\right)
\end{array}
\right)  , \label{bs}%
\end{align}
where%
\begin{equation}
Z\left(  \beta\right)  =\sum_{i}e^{-\beta\varepsilon_{i}} \label{single}%
\end{equation}
is the single-particle partition function of ideal classical gases.\textbf{
}For a free ideal classical gas, the single-particle partition function is
\begin{equation}
Z\left(  \beta\right)  =\frac{V}{\lambda^{3}} \label{single3}%
\end{equation}
\textbf{ }with $V$ the volume and $\lambda=h\sqrt{\frac{\beta}{2\pi m}}$ the
thermal wave length\textbf{ }\cite{Pathria1996Statistical}\textbf{. }

It is worthy to note that different expressions of the $S$-function will give
different expression of canonical partition functions.

\section{The canonical partition function of ideal Fermi gases \label{cFg}}

In this section, we present an exact expression of the canonical partition
function of ideal Fermi gases.

\begin{theorem}
For ideal Fermi gases, the canonical partition function is
\begin{equation}
Z_{FD}\left(  \beta,N\right)  =\left(  1^{N}\right)  \left(  e^{-\beta
\varepsilon_{1}},e^{-\beta\varepsilon_{2}},...\right)  , \label{2222}%
\end{equation}
where $\left(  1^{N}\right)  \left(  x_{1},x_{2},...\right)  $ is the
$S$-function corresponding to the integer partition $\left(  1^{N}\right)  $
defined by Eq. (\ref{ss}).
\end{theorem}

\begin{proof}
Rewrite Eq. (\ref{zzdy}) in terms of occupation numbers. The difference
between Fermi and Bose systems is that the maximum occupation number should be
$1$, i.e., $a_{i}=a_{j}=...a_{k}=1$. Eq. (\ref{zzdy}) then can be re-expressed
as \cite{schrodinger1989statistical}%
\begin{equation}
Z_{FD}\left(  \beta,N\right)  =\sum_{s}e^{-\beta E_{s}}=\sum_{perm}%
e^{-\beta\varepsilon_{i_{1}}-\beta\varepsilon_{i_{2}}.....-\beta
\varepsilon_{i_{N}}}, \label{fdzzdy}%
\end{equation}
where the sum $\sum_{perm}$ runs over all possible mononical increasing
permutation of $\varepsilon_{i}$. The constraint on the total number of
particles of the system now becomes $a_{1}=a_{2}=...=a_{N}=1$.

Equalling the occupation number $a_{i}$ and the element $\lambda_{i}$,
equalling the variable $x_{i}$ and $e^{-\beta\varepsilon_{i}}$, and comparing
the canonical partition function, Eq. (\ref{fdzzdy}), with the definition of
the symmetric function $m_{\left(  \lambda\right)  }\left(  x_{1}%
,x_{2},...\right)  $, Eq. (\ref{mdc}), give
\begin{equation}
Z_{FD}\left(  \beta,N\right)  =m_{\left(  \lambda\right)  _{P\left(  N\right)
}}\left(  e^{-\beta\varepsilon_{1}},e^{-\beta\varepsilon_{2}},...\right)
=m_{\left(  1^{N}\right)  }\left(  e^{-\beta\varepsilon_{1}},e^{-\beta
\varepsilon_{2}},...\right)  , \label{fd mid}%
\end{equation}
where $\left(  \lambda\right)  _{P\left(  N\right)  }=\left(  1^{N}\right)  $
is used. Then we have%
\begin{equation}
Z_{FD}\left(  \beta,N\right)  =\left(  1^{N}\right)  \left(  e^{-\beta
\varepsilon_{1}},e^{-\beta\varepsilon_{2}},...\right)  ,
\end{equation}
where the relation $m_{\left(  1^{N}\right)  }\left(  e^{-\beta\varepsilon
_{1}},e^{-\beta\varepsilon_{2}},...\right)  =\left(  \lambda\right)
_{P\left(  N\right)  }\left(  e^{-\beta\varepsilon_{1}},e^{-\beta
\varepsilon_{2}},...\right)  $ is used \cite{Macdonald1998Symmetric}.
\end{proof}

Similarly, the $S$-function $\left(  1^{N}\right)  \left(  e^{-\beta
\varepsilon_{1}},e^{-\beta\varepsilon_{2}},...\right)  $ can be represented as
the determinant of a certain matrix
\cite{Littlewood2005The,Macdonald1998Symmetric}:
\begin{equation}
\left(  N\right)  \left(  x_{1},x_{2},...\right)  =\frac{1}{N!}\det\left(
\begin{array}
[c]{ccccc}%
\sum_{i}x_{i} & 1 & 0 & ... & 0\\
\sum_{i}x_{i}^{2} & \sum_{i}x_{i} & 2 & ... & 0\\
\sum_{i}x_{i}^{3} & \sum_{i}x_{i}^{2} & \sum_{i}x_{i} & ... & ...\\
... & ... & ... & ... & \left(  N-1\right) \\
\sum_{i}x_{i}^{N} & \sum_{i}x_{i}^{(N-1)} & \sum_{i}x_{i}^{(N-2)} & ... &
\sum_{i}x_{i}%
\end{array}
\right)  . \label{N123}%
\end{equation}
The canonical partition function (\ref{2222}), by Eq. (\ref{N123}), can be
equivalently expressed as
\begin{equation}
Z_{FD}\left(  \beta,N\right)  =\left(  1^{N}\right)  \left(  e^{-\beta
\varepsilon_{1}},e^{-\beta\varepsilon_{2}},...\right)  =\frac{1}{N!}%
\det\left(
\begin{array}
[c]{ccccc}%
Z\left(  \beta\right)  & 1 & 0 & ... & 0\\
Z\left(  2\beta\right)  & Z\left(  \beta\right)  & 2 & ... & 0\\
Z\left(  3\beta\right)  & Z\left(  2\beta\right)  & Z\left(  \beta\right)  &
... & ...\\
... & ... & ... & ... & \left(  N-1\right) \\
Z\left(  N\beta\right)  & Z\left(  N\beta-\beta\right)  & Z\left(
N\beta-2\beta\right)  & ... & Z\left(  \beta\right)
\end{array}
\right)  . \label{fm}%
\end{equation}

\section{The canonical partition function of ideal Gentile gases \label{cGg}}

The Gentile statistic is a generalization of Bose and Fermi statistics
\cite{gentile1940itosservazioni,dai2004gentile}. The maximum occupation number
of a Fermi system is $1$ and of a Bose system is $\infty$. As a
generalization, the maximum occupation number of a Gentile system is an
arbitrary integer $q$
\cite{gentile1940itosservazioni,dai2004gentile,maslov2017relationship}. Beyond
commutative and anticommutative quantization, which corresponds to the Bose
case and the Fermi case respectively, there are also some other effective
quantization schemes \cite{dai2004representation,shen2007intermediate}. It is
shown that the statistical distribution corresponding to various
$q$-deformation schemes are in fact various Gentile distributions with
different maximum occupation numbers $q$ \cite{dai2012calculating}. There are
many physical systems obey intermediate statistics, for example, spin waves,
or, magnons, which is the elementary excitation of the Heisenberg magnetic
system \cite{dai2009intermediate}, deformed fermion gases
\cite{algin2017effective,algin2015anyonic}, and deformed boson gases
\cite{algin2017bose}. Moreover, there are also generalizations for Gentile
statistics, in which the maximum occupation numbers of different quantum
states take on different values \cite{dai2009exactly}.

In this section, we present the canonical partition function of ideal Gentile gases.

\begin{theorem}
For ideal Gentile gases with a maximum occupation number $q$, the canonical
partition function is
\begin{equation}
Z_{q}\left(  \beta,N\right)  =\sum_{I=1}^{P\left(  N\right)  \ }Q^{J}\left(
q\right)  \left(  \lambda\right)  _{I}\left(  e^{-\beta\varepsilon_{1}%
},e^{-\beta\varepsilon_{2}},...\right)  , \label{3333}%
\end{equation}
where the coefficient
\begin{equation}
Q^{J}\left(  q\right)  =\sum_{K=1}^{P\left(  N\right)  }\left(  k_{K}%
^{I}\right)  ^{-1}\Gamma^{K}\left(  q\right)  \label{Q}%
\end{equation}
with $\Gamma^{K}\left(  q\right)  $ satisfying%
\begin{align}
\Gamma^{K}\left(  q\right)   &  =0\text{, \ when }\lambda_{K,1}>q,\nonumber\\
\Gamma^{K}\left(  q\right)   &  =1\text{, \ when }\lambda_{K,1}\leq q,
\label{GAMMA}%
\end{align}
and $\left(  k_{K}^{I}\right)  ^{-1}$ satisfying%
\begin{equation}
\sum_{I=1}^{p\left(  k\right)  }\left(  k_{K}^{I}\right)  ^{-1}k_{I}%
^{L}=\delta_{K}^{L} \label{KOSTKA}%
\end{equation}
with $k_{I}^{L}$ the Kostka number \cite{Macdonald1998Symmetric}.
\end{theorem}

\begin{proof}
Rewrite Eq. (\ref{zzdy}) in terms of the occupation number. Here, the
occupation number $a_{j}$ is restricted by $0<a_{j}\leq q$. Eq. (\ref{zzdy})
can be re-expressed as%
\begin{equation}
Z_{q}\left(  \beta,N\right)  =\sum_{s}e^{-\beta E_{s}}=\sum_{\left\{
a_{i}\right\}  _{q}}\sum_{perm}e^{-\beta\varepsilon_{i_{1}}a_{1}%
-\beta\varepsilon_{i_{2}}a_{2}.....}, \label{fmzzdy1}%
\end{equation}
where the sum $\sum_{perm}$ runs over all possible mononical increasing
permutation of $\varepsilon_{i}$, and the sum $\sum_{\left\{  a_{i}\right\}
_{q}}$ runs over all possible occupation number restricted by the constraints
\begin{equation}
\sum a_{i}=N,\text{ \ \ }1\leq a_{j}\leq q.
\end{equation}

Equalling the occupation number $a_{i}$ and the element $\lambda_{i}$,
equalling the variable $x_{i}$ and $e^{-\beta\varepsilon_{i}}$, and comparing
the canonical partition function, Eq. (\ref{fmzzdy1}), with the definition of
the symmetric function $m_{\left(  \lambda\right)  }\left(  x_{1}%
,x_{2},...\right)  $, Eq. (\ref{mdc}), give%
\begin{equation}
Z_{q}\left(  \beta,N\right)  =\sum_{K=1}^{P\left(  N\right)  \ }\Gamma
^{K}\left(  q\right)  m_{\left(  \lambda\right)  _{K}}\left(  e^{-\beta
\varepsilon_{1}},e^{-\beta\varepsilon_{2}},...\right)  . \label{jtezz}%
\end{equation}
By introducing $\Gamma^{K}\left(  q\right)  $, the constraint on $a_{j}$,
i.e., $1\leq a_{j}\leq q$, is automatically taken into account.

Introducing $\left(  k_{K}^{I}\right)  ^{-1}$, which satisfies Eq.
(\ref{KOSTKA}), multiplying $\left(  k_{K}^{I}\right)  ^{-1}$ to both sides of
Eq. (\ref{smgx}), and summing over the indice $I$, we arrive at%
\begin{equation}
m_{\left(  \lambda\right)  _{K}}\left(  e^{-\beta\varepsilon_{1}}%
,e^{-\beta\varepsilon_{2}},...\right)  =\sum_{I=1}^{P\left(  N\right)
}\left(  k_{K}^{I}\right)  ^{-1}\left(  \lambda\right)  _{I}\left(
e^{-\beta\varepsilon_{1}},e^{-\beta\varepsilon_{2}},...\right)  . \label{gxsm}%
\end{equation}
Substituting Eq. (\ref{gxsm}) into Eq. (\ref{jtezz}) gives
\begin{equation}
Z_{q}\left(  \beta,N\right)  =\sum_{I=1}^{P\left(  N\right)  }\sum
_{K=1}^{P\left(  N\right)  \ }\Gamma^{K}\left(  q\right)  \left(  k_{K}%
^{I}\right)  ^{-1}\left(  \lambda\right)  _{I}\left(  e^{-\beta\varepsilon
_{1}},e^{-\beta\varepsilon_{2}},...\right)  . \label{22222}%
\end{equation}
Substituting $Q^{J}\left(  q\right)  $ which is defined by Eq. (\ref{Q}) into
Eq. (\ref{22222}) gives Eq. (\ref{3333}) directly.
\end{proof}

In appendix \ref{cpf}, as examples, we calculate the canonical partition
function of a Gentile gas based on Eq. (\ref{3333}) for $N=3$, $4$, $5$, and
$6$.

\section{The canonical partition function of interacting classical gases
\label{classical}}

In this section, we show that the canonical partition function of an
interacting classical gas given by the classical cluster expansion is indeed a
Bell polynomial \cite{Andrews1976The,Littlewood2005The,Macdonald1998Symmetric}.

\begin{theorem}
The canonical partition function of an interacting classical gas with $N$
particles is
\begin{equation}
Z\left(  \beta,N\right)  =\frac{1}{N!}B_{N}\left(  \Gamma_{1},\Gamma
_{2},\Gamma_{3},...,\Gamma_{N}\right)  , \label{989898}%
\end{equation}
where $B_{N}\left(  x_{1},x_{2},...,x_{N}\right)  $ is the Bell polynomial
\cite{wheeler1987bell,Collins2001The} and $\Gamma_{l}$ is defined as%
\begin{equation}
\Gamma_{l}=\frac{l!V}{\lambda^{3}}b_{l} \label{Gammal}%
\end{equation}
with $b_{l}$ the expansion coefficient in the classical cluster expansion
\cite{Pathria1996Statistical}.\textbf{ }
\end{theorem}

\begin{proof}
The canonical partition function given by the classical cluster expansion is
\cite{Pathria1996Statistical}%
\begin{equation}
Z\left(  \beta,N\right)  =\sum_{\left\{  m_{l}\right\}  }\left[
{\displaystyle\prod\limits_{l=1}^{N}}
\left(  \frac{1}{l!}\frac{l!b_{l}V}{\lambda^{3}}\right)  ^{m_{l}}\frac
{1}{m_{l}!}\right]  . \label{444444}%
\end{equation}
Comparing Eq. (\ref{444444}) with expression of\ the Bell polynomial
\cite{wheeler1987bell,Collins2001The}%
\begin{equation}
B_{N}\left(  x_{1},x_{2},...,x_{N}\right)  =N!\sum_{\left\{  m_{l}\right\}
}\left[
{\displaystyle\prod\limits_{l=1}^{N}}
\left(  \frac{1}{l!}x_{l}\right)  ^{m_{l}}\frac{1}{m_{l}!}\right]  ,
\label{bell}%
\end{equation}
we immediately arrive at Eq. (\ref{989898}).
\end{proof}

Here, for completeness, we list some $\Gamma_{l}$ in the classical cluster
expansion. For an interacting classical gas with the Hamiltonian $H_{N}%
=-\sum_{l=1}^{N}\frac{P_{l}^{2}}{2m}+\sum_{j<l}U_{jl}$, $\Gamma_{l}$ is given
as \cite{Pathria1996Statistical,ReichlA}%

\begin{align}
\Gamma_{1}  &  =\frac{V}{\lambda^{3}},\nonumber\\
\text{ }\Gamma_{2}  &  =\frac{1}{\lambda^{6}}\int d^{3}q_{1}\int d^{3}%
q_{2}\left(  e^{-\beta U_{12}}-1\right)  ,\cdots.
\end{align}

\section{The canonical partition function of interacting quantum gases
\label{interacting}}

The interacting quantum gas is always an important issue in statistical
mechanics. In a quantum hard-sphere gas, there are two interplayed effects:
quantum exchange interactions and classical inter-particle interactions
\cite{dai2005hard,dai2007interacting,dai2017explicit}. In this section, we
show that the canonical partition function can be represented by the Bell polynomial.

\begin{theorem}
The canonical partition function of an interacting quantum gas with $N$
particles is%
\begin{equation}
Z\left(  \beta,N\right)  =\frac{1}{N!}B_{N}\left(  \Gamma_{1},\Gamma
_{2},\Gamma_{3},...,\Gamma_{N}\right)  , \label{989899}%
\end{equation}
where $B_{N}\left(  x_{1},x_{2},...,x_{N}\right)  $ is the Bell polynomial
\cite{wheeler1987bell,Collins2001The} and $\Gamma_{l}$ is defined as%
\begin{equation}
\Gamma_{l}=\frac{l!V}{\lambda^{3}}b_{l}%
\end{equation}
with $b_{l}$ the expansion coefficient in the quantum cluster expansion
\cite{Pathria1996Statistical,ReichlA}.
\end{theorem}

Here, for completeness, we list some first expansion coefficients in the
quantum cluster expansion. For an interacting quantum gas with the Hamiltonian
$H_{i}=-\frac{\hbar^{2}}{2m}\sum_{l=1}^{i}\nabla_{l}^{2}+\sum_{j<l}U_{jl}$,
the quantum cluster expansion coefficient is given by
\cite{Pathria1996Statistical,ReichlA}\textbf{ }%
\begin{equation}
b_{l}=\frac{1}{l!V\lambda^{3\left(  l-1\right)  }}\int d^{3l}q\langle
q_{1},q_{2},q_{3},...q_{l}|U_{l}|q_{1},q_{2},q_{3},...q_{l}\rangle,
\label{bquantum}%
\end{equation}
with $\langle q_{1}^{\prime}|U_{1}|q_{1}\rangle=\lambda^{3}\langle
q_{1}^{\prime}|e^{-\beta H_{1}}|q_{1}\rangle$ and $\langle q_{1}^{\prime
},q_{2}^{\prime}|U_{2}|q_{1},q_{2}\rangle$=$\lambda^{6}\left[  \langle
q_{1}^{\prime},q_{2}^{\prime}|e^{-\beta H_{2}}|q_{1},q_{2}\rangle\right.
-\left.  \langle q_{1}^{\prime}|e^{-\beta H_{1}}|q_{1}\rangle\langle
q_{2}^{\prime}|e^{-\beta H_{1}}|q_{2}\rangle\right]  $, \ldots. The quantum
exchange interaction is reflected in the symmetric property of the state
vector $|q_{1},q_{2},q_{3},...\rangle$, symmetric for bosons and antisymmetric
for fermions.

\begin{proof}
The canonical partition function given by the quantum cluster expansion is
\cite{Pathria1996Statistical}%
\begin{equation}
Z\left(  \beta,N\right)  =\sum_{\left\{  m_{l}\right\}  }\left[
{\displaystyle\prod\limits_{l=1}^{N}}
\left(  \frac{1}{l!}\frac{l!b_{l}V}{\lambda^{d}}\right)  ^{m_{l}}\frac
{1}{m_{l}!}\right]  . \label{885}%
\end{equation}
Comparing Eq. (\ref{885}) with the expression of the Bell polynomial defined
by Eq. (\ref{bell}) give\ Eq. (\ref{989899}) directly.
\end{proof}

Comparing Eqs. (\ref{989898}) and (\ref{989899}), we can see that the only
difference between the canonical partition functions of an interacting
classical gas and of an interacting quantum gas is the different between the
expansion coefficients.

The above result agrees with the canonical partition function given in Ref.
\cite{fronczak2013cluster}. In Ref. \cite{fronczak2013cluster}, the canonical
partition function is obtained by starting with an expansion of the grand
canonical partition function $\Xi\left(  z,\beta\right)  =\exp\left(
-\beta\Phi\left(  z,\beta\right)  \right)  =\exp\left(  -\beta%
{\displaystyle\sum\limits_{m=1}^{\infty}}
\frac{z^{m}}{m!}\phi_{m}\left(  \beta\right)  \right)  =1+\sum_{N=1}^{\infty
}\frac{z^{N}}{N!}B_{N}\left(  w_{1},...,w_{N}\right)  $ with $w_{l}=-\beta
\phi_{l}\left(  \beta\right)  =-\beta\left.  \frac{\partial^{l}\Phi}{\partial
z^{l}}\right\vert _{z=0}$ and $\Phi\left(  z,\beta\right)  $ the grand
potential, and then comparing with the relation between the grand canonical
partition function and the canonical partition function, $\Xi\left(
z,\beta\right)  =\sum_{N=0}^{\infty}$ $Z\left(  \beta,N\right)  z^{N}$
\cite{Pathria1996Statistical}. In Ref. \cite{fronczak2013cluster}, the
canonical partition function is obtained from a grand partition function. In
the present paper, the canonical partition function is calculated in the
canonical ensemble. Concretely, in this paper, the canonical partition
function for interacting classical and quantum gases is obtained by comparing
the canonical partition function of interacting classical and quantum gases
given by the cluster expansion method \cite{Pathria1996Statistical} and the
definition of the Bell polynomial \cite{wheeler1987bell}. This allows us to
give the parameter $\Gamma_{l}$ in the canonical partition function explicitly.

\section{The virial coefficients of ideal quantum gases in the canonical
ensemble \label{virial}}

Based on the canonical partition function of ideal quantum gases given in
sections \ref{cBg}, \ref{cFg}, and \ref{cGg}, we calculate the virial
coefficient for ideal quantum gases and then compare them with those
calculated in the grand canonical ensemble.

\subsection{The virial coefficients of ideal Bose and Fermi gases in the
canonical ensemble}

In this section we calculate the virial coefficients of ideal Bose and Fermi
gases directly from the canonical partition functions.\textbf{ }

In terms of the $S$-function, the canonical partition functions of ideal Bose
and Fermi gases can be expressed by the partition function of a classical free
particle. Substituting the partition function of a classical free particle
(\ref{single3}) into Eqs. (\ref{bs}) and (\ref{fm}) and re-expressing them in
terms of the $V/\lambda^{3}$ give the canonical partition function of ideal
Bose gases
\begin{equation}
Z_{BE}\left(  \beta,N\right)  =\frac{1}{N!}\left(  \frac{V}{\lambda^{3}%
}\right)  ^{N}+\frac{1}{2^{3/2}}\frac{1}{2\left(  N-2\right)  !}\left(
\frac{V}{\lambda^{3}}\right)  ^{N-1}+\cdots\label{bbs}%
\end{equation}
and the canonical partition function of ideal Fermi gases
\begin{equation}
Z_{FD}\left(  \beta,N\right)  =\frac{1}{N!}\left(  \frac{V}{\lambda^{3}%
}\right)  ^{N}-\frac{1}{2^{3/2}}\frac{1}{2\left(  N-2\right)  !}\left(
\frac{V}{\lambda^{3}}\right)  ^{N-1}+\cdots. \label{ffm}%
\end{equation}

In order to calculate the virial coefficients, we first calculate the equation
of state from the canonical partition function. The equation of state is
\cite{Pathria1996Statistical}
\begin{equation}
P=\frac{1}{\beta}\frac{\partial\ln Z\left(  N,\beta\right)  }{\partial V}.
\label{xtfc}%
\end{equation}
The equation of state \textbf{(}\ref{xtfc}) can be expressed by the virial
expansion \cite{Pathria1996Statistical}:%
\begin{equation}
\frac{PV}{NkT}=%
{\displaystyle\sum\limits_{n}}
a_{n}\left(  \frac{N\lambda^{3}}{V}\right)  ^{n-1} \label{zt}%
\end{equation}
with $a_{l}$ the virial coefficient \cite{Pathria1996Statistical}.

Substituting Eqs. (\ref{bbs}) and (\ref{ffm}) into Eq. (\ref{xtfc}), the
coefficients of $\left(  N\lambda^{3}/V\right)  ^{2}$ give the second virial
coefficients of ideal Bose and Fermi gases, respectively:
\begin{equation}
a_{2}=\pm\frac{1}{2^{5/2}}\left(  1-\frac{1}{N}\right)  , \label{bsfm}%
\end{equation}
where "$+$" stands for Bose gases and "$-$" stands for Fermi gases.

Comparing the second virial coefficient obtained in the canonical ensemble,
Eq. (\ref{bsfm}) with the second virial coefficient obtained in the grand
canonical ensemble \cite{Pathria1996Statistical},%
\begin{equation}
a_{2}^{grand}=\pm\frac{1}{2^{5/2}},
\end{equation}
we can see that the second virial coefficient in the canonical partition
function depends on the total particle number $N$. As $N$ tends to infinity,
$a_{2}$ recovers the second virial coefficient obtained in the grand canonical
ensemble $a_{2}^{grand}$.

\subsection{The virial coefficients of Gentile gases in the canonical
ensemble}

For Gentile gases, the maximum occupation number is an given integer $q$. In
this section, we list some virial coefficients of Gentle gases. The details of
the calculation will be given in the Appendix \ref{QIq}.

In the following, we compare the results obtained in canonical ensembles and
in grand canonical ensembles in Tables (\ref{table1})-(\ref{table5}),
respectively. The virial coefficient of Gentile gases in the grand canonical
ensemble can be found in Ref. \cite{khare2005fractional}.

\begin{table}[ptb]
\caption{The virial coefficients of ideal Gentile gases: $q=2$}%
\label{table1}
\centering
\begin{tabular}
[c]{cccccc}\hline
& $N=3$ & $N=4$ & $N=5$ & $N=6$ & \text{grand canonical}\\\hline
$a_{1}$ & $1$ & $1$ & $1$ & $1$ & $1$\\
$a_{2}$ & $-0.11785$ & $-0.13258$ & $-0.14142$ & $-0.14731$ & $-0.17677$\\
$a_{3}$ & $0.09869$ & $0.15482$ & $0.19317$ & $0.22068$ & $0.386$\\
$a_{4}$ & $-0.04497$ & $-0.11330$ & $-0.17447$ & $-0.22482$ & $-0.6124$%
\\\hline
\end{tabular}
\end{table}

\begin{table}[ptb]
\caption{The virial coefficients of ideal Gentile gases: $q=3$}%
\label{table2}
\centering
\begin{tabular}
[c]{cccccc}\hline
& $N=3$ & $N=4$ & $N=5$ & $N=6$ & \text{grand canonical}\\\hline
$a_{1}$ & $1$ & $1$ & $1$ & $1$ & $1$\\
$a_{2}$ & $-0.11785$ & $-0.13258$ & $-0.14142$ & $-0.14731$ & $-0.17677$\\
$a_{3}$ & $0.01316$ & $0.01048$ & $0.00842$ & $0.00685$ & $-0.0033$\\
$a_{4}$ & $0.00039$ & $0.03667$ & $0.07390$ & $0.10616$ & $0.374889$\\\hline
\end{tabular}
\end{table}

\begin{table}[ptb]
\caption{The virial coefficients of ideal Gentile gases: $q=4$}%
\label{table3}
\centering
\begin{tabular}
[c]{ccccc}\hline
& $N=4$ & $N=5$ & $N=6$ & \text{grand canonical}\\\hline
$a_{1}$ & $1$ & $1$ & $1$ & $1$\\
$a_{2}$ & $-0.13258$ & $-0.14142$ & $-0.14731$ & $-0.17677$\\
$a_{3}$ & $0.01048$ & $0.00842$ & $0.00685$ & $-0.0033$\\
$a_{4}$ & $0.00152$ & $0.00190$ & $0.00199$ & $-0.00011$\\\hline
\end{tabular}
\end{table}

\begin{table}[ptb]
\caption{The virial coefficients of ideal Gentile gases: $q=5$}%
\label{table4}
\centering
\begin{tabular}
[c]{cccc}\hline
& $N=5$ & $N=6$ & \text{grand canonical}\\\hline
$a_{1}$ & $1$ & $1$ & $1$\\
$a_{2}$ & $-0.14142$ & $-0.14731$ & $-0.17677$\\
$a_{3}$ & $0.00842$ & $0.00685$ & $-0.0033$\\
$a_{4}$ & $0.00190$ & $0.00199$ & $-0.00011$\\\hline
\end{tabular}
\end{table}

\begin{table}[ptb]
\caption{The virial coefficients of ideal Gentile gases: $q=6$}%
\label{table5}
\centering
\begin{tabular}
[c]{ccc}\hline
& $N=6$ & \text{grand canonical}\\\hline
$a_{1}$ & $1$ & $1$\\
$a_{2}$ & $-0.14731$ & $-0.17677$\\
$a_{3}$ & $0.00685$ & $-0.0033$\\
$a_{4}$ & $0.0020$ & $-0.00011$\\\hline
\end{tabular}
\end{table}

From Tables (\ref{table1})-(\ref{table5}), one can see that the virial
coefficients calculated from in the canonical ensemble are different from
those calculated in the grand canonical ensemble. However, as $N$ increases,
the result obtained in the canonical ensemble tends to the result obtained in
the grand canonical ensemble.

\section{The virial coefficients of interacting classical and quantum gases in
the canonical ensemble \label{virialint}}

In this section, based on the expansion of the Bell polynomial, we calculate
the second virial coefficient of interacting classical and quantum gases in
canonical ensemble and then compare the results with those calculated in the
grand canonical ensemble.

The Bell polynomial can be expended as%
\begin{equation}
B_{N}\left[  \Gamma_{1},\Gamma_{2},...,\Gamma_{N}\right]  =\Gamma_{1}%
^{N}+\frac{N!}{2\left(  N-2\right)  !}\Gamma_{1}^{N-2}\Gamma_{2}%
+...+\Gamma_{N}.
\end{equation}
Then the canonical partition function (\ref{989898}) becomes%
\begin{equation}
Z\left(  \beta,N\right)  =\frac{1}{N!}\Gamma_{1}^{N}+\frac{1}{2\left(
N-2\right)  !}\Gamma_{1}^{N-2}\Gamma_{2}+...+\frac{1}{N!}\Gamma_{N}.
\label{zk}%
\end{equation}

The second virial coefficients of interacting classical and quantum gases can
be directly obtained from the canonical partition function. By Eq.
(\ref{Gammal}), we have $\Gamma_{1}=\frac{V}{\lambda^{3}}$ and $\Gamma
_{2}=\frac{2V}{\lambda^{3}}b_{2}$. Then Eq. (\ref{zk}) becomes%
\begin{equation}
Z\left(  \beta,N\right)  =\frac{1}{N!}\left(  \frac{V}{\lambda^{3}}\right)
^{N}+\frac{1}{\left(  N-2\right)  !}\left(  \frac{V}{\lambda^{3}}\right)
^{N-1}b_{2}+\cdots. \label{zks}%
\end{equation}
Substituting Eq. (\ref{zks}) into Eq. (\ref{xtfc}) and selecting the
coefficients of $\left(  N\lambda^{3}/V\right)  ^{2}$ give the second virial
coefficients:%
\begin{equation}
a_{2}=b_{2}\left(  1-\frac{1}{N}\right)  . \label{sdsd}%
\end{equation}
One can see that the second virial coefficients are functions of the total
particle number $N$.

Comparing the second virial coefficient obtained in the canonical ensemble,
Eq. (\ref{sdsd}), with the second virial coefficient obtained in the grand
canonical ensemble \cite{Pathria1996Statistical},%
\begin{equation}
a_{2}^{grand}=b_{2},
\end{equation}
we can see that the result obtained in the canonical partition function
depends on the total particle number $N$ and recovers the result obtained in
the grand canonical ensemble when $N$ tends to infinite.

\section{Conclusions and discussions \label{Conclusionoutlook}}

Solving an $N$-body system is an important problem in physics. There are many
mechanical and statistical-mechanical methods developed to this problem.

The mechanical treatment for an $N$-body system is to find the eigenvalues and
eigenstates of the $N$-body Hamiltonian of the system. In the mechanical
treatment, both informations of eigenvalues and eigenfunctions are taken into
account. However, as is well known, solving an $N$-body system in mechanics is
very difficult. In an interacting quantum system, two factors, classical
inter-particle interactions and the quantum exchange interactions, intertwined together.

The statistical-mechanical treatment for an $N$-body system, essentially, is
to use average values instead of exact values. Only the information of
eigenvalues is taken into account. All thermodynamic quantities can be
obtained from, for example, the canonical partition function $Z\left(
\beta,N\right)  =\sum_{s}e^{-\beta E_{s}}$ which is determined only by the
eigenvalues $E_{s}$ of the $N$-body system. At the expense of losing the
information of eigenfunctions, the statistical-mechanical treatment is much
easier compared with the mechanical treatment. In a word, the
statistical-mechanical treatment can be regarded as an approximate method for
$N$-body problems.

A rigorous description of an $N$-body system in statistical mechanics should
be made in a canonical ensemble with fixed particle number $N$.\ The core task
in a canonical ensemble is to seek the canonical partition function $Z\left(
\beta,N\right)  $. The calculation of the canonical partition function is in
fact a sum over all states with a constraint of the total particle number of
the system. Nevertheless, it is very difficult to calculate the canonical
partition function when one has to deal with the difficulties from the
classical inter-particle interactions and quantum exchange interactions and
take the constraint of the total particle number of the system into
consideration simultaneously.

In order to avoid the difficulty caused by the constraint of the total
particle number, conventionally, one introduces the grand canonical ensemble
which\ has no constraint on the particle number, and, instead of canonical
partition function, one turns to calculate the grand canonical partition
function. The advantage is that the constraint of particle numbers on the sum
is now removed; the price is that the particle number now is not a constant.
In the grand canonical ensemble, the total particle number is not rigorously
equal to an exact number $N$, but is equal to the average particle number
$\left\langle N\right\rangle $. That is to say, an $N$-body problem is
approximately converted to an $\left\langle N\right\rangle $-body problem,
with a deviation of $\frac{1}{\sqrt{N}}$. Unless in the thermodynamic limit,
i.e., $N\rightarrow\infty$, the result obtained in the grand canonical
ensemble is not the same as that obtained in the canonical ensemble. Thus, for
an $N$-body problem, the grand canonical ensemble is an approximate method
even in statistical-mechanical treatment.

The method suggested in the paper is based on the mathematical theory of the
symmetric function and the Bell polynomial. In this paper, for ideal quantum
gases, including Bose, Fermi, and Gentile gases, we suggest a method to
calculate exact canonical partition functions, and show that the canonical
partition functions of ideal Bose, Fermi, and Gentile gases can be represented
as linear combinations of the $S$-function. For interacting classical and
quantum gases, we point out that the canonical partition functions given by
the cluster expansion method are indeed the Bell polynomial. Starting from the
exact canonical partition functions of ideal quantum gases, we calculate the
virial coefficients in canonical ensemble. For interacting gases, we calculate
the virial coefficients in canonical ensembles instead of in grand canonical
ensembles as that in literature.
\appendix

\section{Appendix: Details of the calculation for ideal Gentile gases in the
canonical ensemble \label{appendix}}

In this appendix, we give some details for the calculation of the ideal
Gentile gases in the canonical ensemble.

\subsection{The coefficient $Q^{I}\left(  q\right)  $ \label{QIq}}

The canonical partition function of ideal Gentile gases, Eq. (\ref{3333}), can
be represented as a linear combination of the $S$-function $\left(
\lambda\right)  _{I}\left(  e^{-\beta\varepsilon_{1}},e^{-\beta\varepsilon
_{2}},...\right)  $ and the corresponding coefficient $Q^{I}\left(  q\right)
$ is defined by Eq. (\ref{Q}). That is to say, the canonical partition
function of ideal Gentile gases will be obtained once the coefficient
$Q^{I}\left(  q\right)  $ is known. In this Appendix, we show how to calculate
the coefficient $Q^{I}\left(  q\right)  $ for some given $N$.

$N=3$\textit{.} The Kostka number $k_{K}^{J}$ for $N=3$ reads $k_{\left(
3\right)  }^{\left(  3\right)  }=k_{1}^{1}=1$, $k_{\left(  2,1\right)
}^{\left(  3\right)  }=k_{2}^{1}=0$, $k_{\left(  1^{3}\right)  }^{\left(
3\right)  }=k_{3}^{1}=0$, $k_{\left(  3\right)  }^{\left(  2,1\right)  }%
=k_{1}^{2}=1$, $k_{\left(  2,1\right)  }^{\left(  2,1\right)  }=k_{2}^{2}=1$,
$k_{\left(  1^{3}\right)  }^{\left(  2,1\right)  }=k_{3}^{2}=0$, $k_{\left(
3\right)  }^{\left(  1^{3}\right)  }=k_{1}^{3}=1$, $k_{\left(  2,1\right)
}^{\left(  1^{3}\right)  }=k_{2}^{3}=2$, and $k_{\left(  1^{3}\right)
}^{\left(  1^{3}\right)  }=k_{3}^{3}=1$
\cite{Littlewood2005The,Macdonald1998Symmetric}. For clarity , we can rewrite
the Kostka number in a matrix form:%
\begin{equation}
k=\left(
\begin{array}
[c]{ccc}%
1 & 0 & 0\\
1 & 1 & 0\\
1 & 2 & 1
\end{array}
\right)  , \label{a1}%
\end{equation}
where we take the upper index of $k_{K}^{J}$ as a row index and the lower
index as a column index.

For $q=3$, by the definition of $\Gamma^{K}\left(  q\right)  $, Eq.
(\ref{GAMMA}), we have $\Gamma_{\left(  3\right)  }\left(  3\right)
=\Gamma_{1}\left(  3\right)  =1$, $\Gamma_{\left(  2,1\right)  }\left(
3\right)  =\Gamma_{2}\left(  3\right)  =1$, and $\Gamma_{\left(  1^{3}\right)
}\left(  3\right)  =\Gamma_{3}\left(  3\right)  =1$. For clarity, we express
$\Gamma^{K}\left(  q\right)  $ as%
\begin{equation}
\Gamma\left(  3\right)  =\left(
\begin{array}
[c]{c}%
1\\
1\\
1
\end{array}
\right)  , \label{a2}%
\end{equation}
where we take the upper index of $\Gamma^{K}\left(  q\right)  $ as a row
index. Again we express Eq. (\ref{Q}) in the matrix form
\begin{equation}
Q\left(  q\right)  =k^{-1}\Gamma\left(  q\right)  , \label{Qmatrix}%
\end{equation}
where $Q\left(  q\right)  $ denotes the vector form of $Q^{I}\left(  q\right)
$ by taking the upper index as a row index, and $k^{-1}$ is the matrix form of
$\left(  k^{-1}\right)  _{K}^{J}$. Then substituting Eqs. (\ref{a1}) and
(\ref{a2}) into Eq. (\ref{Qmatrix}) gives%
\begin{equation}
Q\left(  3\right)  =\left(
\begin{array}
[c]{ccc}%
1 & 0 & 0\\
1 & 1 & 0\\
1 & 2 & 1
\end{array}
\right)  ^{-1}\left(
\begin{array}
[c]{c}%
1\\
1\\
1
\end{array}
\right)  =\left(
\begin{array}
[c]{c}%
1\\
0\\
0
\end{array}
\right)  .
\end{equation}
That is to say $Q^{\left(  3\right)  }\left(  3\right)  =Q^{1}\left(
3\right)  =1$, $Q^{\left(  2,1\right)  }\left(  3\right)  =Q^{2}\left(
3\right)  =0$, $Q^{\left(  1^{3}\right)  }\left(  3\right)  =Q^{3}\left(
3\right)  =0$.

Therefore, when $N=3$ and $q=3$, Eq. (\ref{3333}) can be expressed as
\begin{align}
Z_{3}\left(  \beta,3\right)   &  =1\times\left(  3\right)  \left(
e^{-\beta\varepsilon_{1}},e^{-\beta\varepsilon_{2}},...\right)  +0\times
\left(  2,1\right)  \left(  e^{-\beta\varepsilon_{1}},e^{-\beta\varepsilon
_{2}},...\right)  +0\times\left(  1^{3}\right)  \left(  e^{-\beta
\varepsilon_{1}},e^{-\beta\varepsilon_{2}},...\right) \nonumber\\
&  =\left(  3\right)  \left(  e^{-\beta\varepsilon_{1}},e^{-\beta
\varepsilon_{2}},...\right)  , \label{PARTITION2}%
\end{align}
where $\left(  \lambda\right)  _{1}\left(  e^{-\beta\varepsilon_{1}}%
,e^{-\beta\varepsilon_{2}},...\right)  =\left(  3\right)  \left(
e^{-\beta\varepsilon_{1}},e^{-\beta\varepsilon_{2}},...\right)  $, $\left(
\lambda\right)  _{2}\left(  e^{-\beta\varepsilon_{1}},e^{-\beta\varepsilon
_{2}},...\right)  =\left(  2,1\right)  \left(  e^{-\beta\varepsilon_{1}%
},e^{-\beta\varepsilon_{2}},...\right)  $, and $\left(  \lambda\right)
_{3}\left(  e^{-\beta\varepsilon_{1}},e^{-\beta\varepsilon_{2}},...\right)
=\left(  1^{3}\right)  \left(  e^{-\beta\varepsilon_{1}},e^{-\beta
\varepsilon_{2}},...\right)  $. Eq. (\ref{PARTITION2}) is just the canonical
partition function for Gentile gases with the particle number $N=3$ and\ the
maximum occupation number $q=3$.

For $q=2$, by Eq. (\ref{GAMMA}), we have%
\begin{equation}
\Gamma\left(  2\right)  =\left(
\begin{array}
[c]{c}%
0\\
1\\
1
\end{array}
\right)  . \label{a3}%
\end{equation}
Substituting Eqs. (\ref{a3}) and (\ref{a1}) into Eq. (\ref{Qmatrix}) gives%
\begin{equation}
Q\left(  2\right)  =\left(
\begin{array}
[c]{ccc}%
1 & 0 & 0\\
1 & 1 & 0\\
1 & 2 & 1
\end{array}
\right)  ^{-1}\left(
\begin{array}
[c]{c}%
0\\
1\\
1
\end{array}
\right)  =\left(
\begin{array}
[c]{c}%
0\\
1\\
-1
\end{array}
\right)  . \label{Q32}%
\end{equation}
Substituting the coefficient $Q\left(  2\right)  $, Eq. (\ref{Q32}), into Eq.
(\ref{3333}) gives the canonical partition function\textbf{ }%
\begin{align}
Z_{2}\left(  \beta,3\right)   &  =0\times\left(  3\right)  \left(
e^{-\beta\varepsilon_{1}},e^{-\beta\varepsilon_{2}},...\right)  +1\times
\left(  2,1\right)  \left(  e^{-\beta\varepsilon_{1}},e^{-\beta\varepsilon
_{2}},...\right)  -1\times\left(  1^{3}\right)  \left(  e^{-\beta
\varepsilon_{1}},e^{-\beta\varepsilon_{2}},...\right) \nonumber\\
&  =\left(  2,1\right)  \left(  e^{-\beta\varepsilon_{1}},e^{-\beta
\varepsilon_{2}},...\right)  -\left(  1^{3}\right)  \left(  e^{-\beta
\varepsilon_{1}},e^{-\beta\varepsilon_{2}},...\right)  . \label{P32}%
\end{align}

For $q=1$, by Eq. (\ref{GAMMA}), we have%
\begin{equation}
\Gamma\left(  1\right)  =\left(
\begin{array}
[c]{c}%
0\\
0\\
1
\end{array}
\right)  . \label{a4}%
\end{equation}
Substituting Eqs. (\ref{a4}) and (\ref{a1}) into Eq. (\ref{Qmatrix}) gives%
\begin{equation}
Q\left(  1\right)  =\left(
\begin{array}
[c]{ccc}%
1 & 0 & 0\\
1 & 1 & 0\\
1 & 2 & 1
\end{array}
\right)  ^{-1}\left(
\begin{array}
[c]{c}%
0\\
0\\
1
\end{array}
\right)  =\left(
\begin{array}
[c]{c}%
0\\
0\\
1
\end{array}
\right)  . \label{Q31}%
\end{equation}
Substituting the coefficient $Q\left(  1\right)  $ given by Eq. (\ref{Q31})
into Eq. (\ref{3333}) gives the canonical partition function\textbf{ }%
\begin{align}
Z_{1}\left(  \beta,3\right)   &  =0\times\left(  3\right)  \left(
e^{-\beta\varepsilon_{1}},e^{-\beta\varepsilon_{2}},...\right)  +0\times
\left(  2,1\right)  \left(  e^{-\beta\varepsilon_{1}},e^{-\beta\varepsilon
_{2}},...\right)  +1\times\left(  1^{3}\right)  \left(  e^{-\beta
\varepsilon_{1}},e^{-\beta\varepsilon_{2}},...\right) \nonumber\\
&  =\left(  1^{3}\right)  \left(  e^{-\beta\varepsilon_{1}},e^{-\beta
\varepsilon_{2}},...\right)  . \label{P31}%
\end{align}

$N=4$\textit{.} The Kostka number for $N=4$ is
\cite{Littlewood2005The,Macdonald1998Symmetric}%
\begin{equation}
k=\left(
\begin{array}
[c]{ccccc}%
1 & 0 & 0 & 0 & 0\\
1 & 1 & 0 & 0 & 0\\
1 & 1 & 1 & 0 & 0\\
1 & 2 & 1 & 1 & 0\\
1 & 3 & 2 & 3 & 1
\end{array}
\right)  . \label{a5}%
\end{equation}

For $q=4$, by Eq. (\ref{GAMMA}), we have
\begin{equation}
\Gamma\left(  4\right)  =\left(
\begin{array}
[c]{c}%
1\\
1\\
1\\
1\\
1
\end{array}
\right)  . \label{a6}%
\end{equation}
Substituting Eqs. (\ref{a5}) and (\ref{a6}) into Eq. (\ref{Qmatrix}) gives
\begin{equation}
Q\left(  4\right)  =\left(
\begin{array}
[c]{ccccc}%
1 & 0 & 0 & 0 & 0\\
1 & 1 & 0 & 0 & 0\\
1 & 1 & 1 & 0 & 0\\
1 & 2 & 1 & 1 & 0\\
1 & 3 & 2 & 3 & 1
\end{array}
\right)  ^{-1}\left(
\begin{array}
[c]{c}%
1\\
1\\
1\\
1\\
1
\end{array}
\right)  =\left(
\begin{array}
[c]{c}%
1\\
0\\
0\\
0\\
0
\end{array}
\right)  . \label{Q44}%
\end{equation}
Substituting the coefficient $Q\left(  4\right)  $ given by Eq. (\ref{Q44})
into Eq. (\ref{3333}) gives the canonical partition function\textbf{ }%
\begin{align}
Z_{4}\left(  4,\beta\right)   &  =1\times\left(  4\right)  \left(
e^{-\beta\varepsilon_{1}},e^{-\beta\varepsilon_{2}},...\right)  +0\times
\left(  3,1\right)  \left(  e^{-\beta\varepsilon_{1}},e^{-\beta\varepsilon
_{2}},...\right)  +0\times\left(  2^{2}\right)  \left(  e^{-\beta
\varepsilon_{1}},e^{-\beta\varepsilon_{2}},...\right) \nonumber\\
&  +0\times\left(  2,1^{2}\right)  \left(  e^{-\beta\varepsilon_{1}}%
,e^{-\beta\varepsilon_{2}},...\right)  +0\times\left(  1^{5}\right)  \left(
e^{-\beta\varepsilon_{1}},e^{-\beta\varepsilon_{2}},...\right) \nonumber\\
&  =\left(  4\right)  \left(  e^{-\beta\varepsilon_{1}},e^{-\beta
\varepsilon_{2}},...\right)  . \label{P44}%
\end{align}

For $q=3$, by Eq. (\ref{GAMMA}), we have
\begin{equation}
\Gamma\left(  3\right)  =\left(
\begin{array}
[c]{c}%
0\\
1\\
1\\
1\\
1
\end{array}
\right)  . \label{a7}%
\end{equation}
Substituting Eqs. (\ref{a5}) and (\ref{a7}) into Eq. (\ref{Qmatrix}) gives
\begin{equation}
Q\left(  3\right)  =\left(
\begin{array}
[c]{ccccc}%
1 & 0 & 0 & 0 & 0\\
1 & 1 & 0 & 0 & 0\\
1 & 1 & 1 & 0 & 0\\
1 & 2 & 1 & 1 & 0\\
1 & 3 & 2 & 3 & 1
\end{array}
\right)  ^{-1}\left(
\begin{array}
[c]{c}%
0\\
1\\
1\\
1\\
1
\end{array}
\right)  =\left(
\begin{array}
[c]{c}%
0\\
1\\
0\\
-1\\
1
\end{array}
\right)  . \label{Q43}%
\end{equation}
Substituting the coefficient $Q\left(  3\right)  $ given by (\ref{Q43}) into
Eq. (\ref{3333}) gives the canonical partition function\textbf{ }%
\begin{align}
Z_{3}\left(  \beta,4\right)   &  =0\times\left(  4\right)  \left(
e^{-\beta\varepsilon_{1}},e^{-\beta\varepsilon_{2}},...\right)  +1\times
\left(  3,1\right)  \left(  e^{-\beta\varepsilon_{1}},e^{-\beta\varepsilon
_{2}},...\right)  +0\times\left(  2^{2}\right)  \left(  e^{-\beta
\varepsilon_{1}},e^{-\beta\varepsilon_{2}},...\right) \nonumber\\
&  -1\times\left(  2,1^{2}\right)  \left(  e^{-\beta\varepsilon_{1}}%
,e^{-\beta\varepsilon_{2}},...\right)  +1\times\left(  1^{5}\right)  \left(
e^{-\beta\varepsilon_{1}},e^{-\beta\varepsilon_{2}},...\right) \nonumber\\
&  =\left(  3,1\right)  \left(  e^{-\beta\varepsilon_{1}},e^{-\beta
\varepsilon_{2}},...\right)  -\left(  2,1^{2}\right)  \left(  e^{-\beta
\varepsilon_{1}},e^{-\beta\varepsilon_{2}},...\right)  +\left(  1^{5}\right)
\left(  e^{-\beta\varepsilon_{1}},e^{-\beta\varepsilon_{2}},...\right)  .
\label{P43}%
\end{align}

For $q=2$, by Eq. (\ref{GAMMA}), we have%
\begin{equation}
\Gamma\left(  2\right)  =\left(
\begin{array}
[c]{c}%
0\\
0\\
1\\
1\\
1
\end{array}
\right)  . \label{a8}%
\end{equation}
Substituting Eqs. (\ref{a5}) and (\ref{a8}) into Eq. (\ref{Qmatrix}) gives%
\begin{equation}
Q\left(  2\right)  =\left(
\begin{array}
[c]{ccccc}%
1 & 0 & 0 & 0 & 0\\
1 & 1 & 0 & 0 & 0\\
1 & 1 & 1 & 0 & 0\\
1 & 2 & 1 & 1 & 0\\
1 & 3 & 2 & 3 & 1
\end{array}
\right)  ^{-1}\left(
\begin{array}
[c]{c}%
0\\
0\\
1\\
1\\
1
\end{array}
\right)  =\left(
\begin{array}
[c]{c}%
0\\
0\\
1\\
0\\
-1
\end{array}
\right)  . \label{Q42}%
\end{equation}
Substituting the coefficient $Q\left(  2\right)  $ given by Eq. (\ref{Q42})
into Eq. (\ref{3333}) gives the canonical partition function\textbf{ }%
\begin{align}
Z_{2}\left(  \beta,4\right)   &  =0\times\left(  4\right)  \left(
e^{-\beta\varepsilon_{1}},e^{-\beta\varepsilon_{2}},...\right)  +0\times
\left(  3,1\right)  \left(  e^{-\beta\varepsilon_{1}},e^{-\beta\varepsilon
_{2}},...\right)  +1\times\left(  2^{2}\right)  \left(  e^{-\beta
\varepsilon_{1}},e^{-\beta\varepsilon_{2}},...\right) \nonumber\\
&  +0\times\left(  2,1^{2}\right)  \left(  e^{-\beta\varepsilon_{1}}%
,e^{-\beta\varepsilon_{2}},...\right)  -1\times\left(  1^{5}\right)  \left(
e^{-\beta\varepsilon_{1}},e^{-\beta\varepsilon_{2}},...\right) \nonumber\\
&  =\left(  2^{2}\right)  \left(  e^{-\beta\varepsilon_{1}},e^{-\beta
\varepsilon_{2}},...\right)  -\left(  1^{5}\right)  \left(  e^{-\beta
\varepsilon_{1}},e^{-\beta\varepsilon_{2}},...\right)  . \label{P42}%
\end{align}

For $q=1$, by Eq. (\ref{GAMMA}), we have%
\begin{equation}
\Gamma\left(  1\right)  =\left(
\begin{array}
[c]{c}%
0\\
0\\
0\\
0\\
1
\end{array}
\right)  . \label{a9}%
\end{equation}
Substituting Eqs. (\ref{a5}) and (\ref{a9}) into Eq. (\ref{Qmatrix}) gives%
\begin{equation}
Q\left(  1\right)  =\left(
\begin{array}
[c]{ccccc}%
1 & 0 & 0 & 0 & 0\\
1 & 1 & 0 & 0 & 0\\
1 & 1 & 1 & 0 & 0\\
1 & 2 & 1 & 1 & 0\\
1 & 3 & 2 & 3 & 1
\end{array}
\right)  ^{-1}\left(
\begin{array}
[c]{c}%
0\\
0\\
0\\
0\\
1
\end{array}
\right)  =\left(
\begin{array}
[c]{c}%
0\\
0\\
0\\
0\\
1
\end{array}
\right)  . \label{Q41}%
\end{equation}
Substituting the coefficient $Q\left(  2\right)  $ given by Eq. (\ref{Q42})
into Eq. (\ref{3333}) gives the canonical partition function\textbf{ }%
\begin{align}
Z_{1}\left(  \beta,4\right)   &  =0\times\left(  4\right)  \left(
e^{-\beta\varepsilon_{1}},e^{-\beta\varepsilon_{2}},...\right)  +0\times
\left(  3,1\right)  \left(  e^{-\beta\varepsilon_{1}},e^{-\beta\varepsilon
_{2}},...\right)  +0\times\left(  2^{2}\right)  \left(  e^{-\beta
\varepsilon_{1}},e^{-\beta\varepsilon_{2}},...\right) \nonumber\\
&  +0\times\left(  2,1^{2}\right)  \left(  e^{-\beta\varepsilon_{1}}%
,e^{-\beta\varepsilon_{2}},...\right)  +1\times\left(  1^{5}\right)  \left(
e^{-\beta\varepsilon_{1}},e^{-\beta\varepsilon_{2}},...\right) \nonumber\\
&  =\left(  1^{5}\right)  \left(  e^{-\beta\varepsilon_{1}},e^{-\beta
\varepsilon_{2}},...\right)  . \label{P41}%
\end{align}

$N=5$. The Kostka number for $N=5$ is
\cite{Littlewood2005The,Macdonald1998Symmetric}%
\begin{equation}
k=\left(
\begin{array}
[c]{ccccccc}%
1 & 0 & 0 & 0 & 0 & 0 & 0\\
1 & 1 & 0 & 0 & 0 & 0 & 0\\
1 & 1 & 1 & 0 & 0 & 0 & 0\\
1 & 2 & 1 & 1 & 0 & 0 & 0\\
1 & 2 & 2 & 1 & 1 & 0 & 0\\
1 & 3 & 3 & 3 & 2 & 1 & 0\\
1 & 4 & 5 & 6 & 5 & 4 & 1
\end{array}
\right)  . \label{a10}%
\end{equation}

For $q=5$, by Eq. (\ref{GAMMA}), we have%
\begin{equation}
\Gamma\left(  5\right)  =\left(
\begin{array}
[c]{c}%
1\\
1\\
1\\
1\\
1\\
1\\
1
\end{array}
\right)  . \label{a11}%
\end{equation}
Substituting Eqs. (\ref{a10}) and (\ref{a11}) into Eq. (\ref{Qmatrix}) gives%
\begin{equation}
Q\left(  5\right)  =\left(
\begin{array}
[c]{ccccccc}%
1 & 0 & 0 & 0 & 0 & 0 & 0\\
1 & 1 & 0 & 0 & 0 & 0 & 0\\
1 & 1 & 1 & 0 & 0 & 0 & 0\\
1 & 2 & 1 & 1 & 0 & 0 & 0\\
1 & 2 & 2 & 1 & 1 & 0 & 0\\
1 & 3 & 3 & 3 & 2 & 1 & 0\\
1 & 4 & 5 & 6 & 5 & 4 & 1
\end{array}
\right)  ^{-1}\left(
\begin{array}
[c]{c}%
1\\
1\\
1\\
1\\
1\\
1\\
1
\end{array}
\right)  =\left(
\begin{array}
[c]{c}%
1\\
0\\
0\\
0\\
0\\
0\\
0
\end{array}
\right)  . \label{Q55}%
\end{equation}
Substituting the coefficients $Q\left(  5\right)  $ given by Eq. (\ref{Q55})
into Eq. (\ref{3333}) gives the canonical partition function\textbf{ }%
\begin{align}
Z_{5}\left(  \beta,5\right)   &  =1\times\left(  5\right)  \left(
e^{-\beta\varepsilon_{1}},e^{-\beta\varepsilon_{2}},...\right)  +0\times
\left(  4,1\right)  \left(  e^{-\beta\varepsilon_{1}},e^{-\beta\varepsilon
_{2}},...\right)  +0\times\left(  3,2\right)  \left(  e^{-\beta\varepsilon
_{1}},e^{-\beta\varepsilon_{2}},...\right) \nonumber\\
&  +0\times\left(  3,1^{2}\right)  \left(  e^{-\beta\varepsilon_{1}}%
,e^{-\beta\varepsilon_{2}},...\right)  +0\times\left(  2^{2},1\right)  \left(
e^{-\beta\varepsilon_{1}},e^{-\beta\varepsilon_{2}},...\right)  +0\times
\left(  2,1^{3}\right)  \left(  e^{-\beta\varepsilon_{1}},e^{-\beta
\varepsilon_{2}},...\right) \nonumber\\
&  +0\times\left(  1^{5}\right)  \left(  e^{-\beta\varepsilon_{1}}%
,e^{-\beta\varepsilon_{2}},...\right) \nonumber\\
&  =\left(  5\right)  \left(  e^{-\beta\varepsilon_{1}},e^{-\beta
\varepsilon_{2}},...\right)  . \label{P55}%
\end{align}

For $q=4$, by Eq. (\ref{GAMMA}), we have
\begin{equation}
\Gamma\left(  4\right)  =\left(
\begin{array}
[c]{c}%
0\\
1\\
1\\
1\\
1\\
1\\
1
\end{array}
\right)  . \label{a12}%
\end{equation}
Substituting Eqs. (\ref{a10}) and (\ref{a12}) into Eq. (\ref{Qmatrix}) gives%
\begin{equation}
Q\left(  4\right)  =\left(
\begin{array}
[c]{ccccccc}%
1 & 0 & 0 & 0 & 0 & 0 & 0\\
1 & 1 & 0 & 0 & 0 & 0 & 0\\
1 & 1 & 1 & 0 & 0 & 0 & 0\\
1 & 2 & 1 & 1 & 0 & 0 & 0\\
1 & 2 & 2 & 1 & 1 & 0 & 0\\
1 & 3 & 3 & 3 & 2 & 1 & 0\\
1 & 4 & 5 & 6 & 5 & 4 & 1
\end{array}
\right)  ^{-1}\left(
\begin{array}
[c]{c}%
0\\
1\\
1\\
1\\
1\\
1\\
1
\end{array}
\right)  =\left(
\begin{array}
[c]{c}%
0\\
1\\
0\\
-1\\
0\\
1\\
-1
\end{array}
\right)  . \label{Q54}%
\end{equation}
Substituting the coefficient $Q\left(  4\right)  $ given by Eq. (\ref{Q54})
into Eq. (\ref{3333}) gives the canonical partition function\textbf{ }%
\begin{align}
Z_{4}\left(  \beta,5\right)   &  =0\times\left(  5\right)  \left(
e^{-\beta\varepsilon_{1}},e^{-\beta\varepsilon_{2}},...\right)  +1\times
\left(  4,1\right)  \left(  e^{-\beta\varepsilon_{1}},e^{-\beta\varepsilon
_{2}},...\right)  +0\times\left(  3,2\right)  \left(  e^{-\beta\varepsilon
_{1}},e^{-\beta\varepsilon_{2}},...\right) \nonumber\\
&  -1\times\left(  3,1^{2}\right)  \left(  e^{-\beta\varepsilon_{1}}%
,e^{-\beta\varepsilon_{2}},...\right)  +0\times\left(  2^{2},1\right)  \left(
e^{-\beta\varepsilon_{1}},e^{-\beta\varepsilon_{2}},...\right)  +1\times
\left(  2,1^{3}\right)  \left(  e^{-\beta\varepsilon_{1}},e^{-\beta
\varepsilon_{2}},...\right) \nonumber\\
&  -1\times\left(  1^{5}\right)  \left(  e^{-\beta\varepsilon_{1}}%
,e^{-\beta\varepsilon_{2}},...\right) \nonumber\\
&  =\left(  4,1\right)  \left(  e^{-\beta\varepsilon_{1}},e^{-\beta
\varepsilon_{2}},...\right)  -\left(  3,1^{2}\right)  \left(  e^{-\beta
\varepsilon_{1}},e^{-\beta\varepsilon_{2}},...\right) \nonumber\\
&  +\left(  2,1^{3}\right)  \left(  e^{-\beta\varepsilon_{1}},e^{-\beta
\varepsilon_{2}},...\right)  -\left(  1^{5}\right)  \left(  e^{-\beta
\varepsilon_{1}},e^{-\beta\varepsilon_{2}},...\right)  . \label{P54}%
\end{align}

For $q=3$, by Eq. (\ref{GAMMA}), we have%
\begin{equation}
\Gamma\left(  3\right)  =\left(
\begin{array}
[c]{c}%
0\\
0\\
1\\
1\\
1\\
1\\
1
\end{array}
\right)  . \label{a13}%
\end{equation}
Substituting Eqs. (\ref{a10}) and (\ref{a13}) into Eq. (\ref{Qmatrix}) gives%
\begin{equation}
Q\left(  3\right)  =\left(
\begin{array}
[c]{ccccccc}%
1 & 0 & 0 & 0 & 0 & 0 & 0\\
1 & 1 & 0 & 0 & 0 & 0 & 0\\
1 & 1 & 1 & 0 & 0 & 0 & 0\\
1 & 2 & 1 & 1 & 0 & 0 & 0\\
1 & 2 & 2 & 1 & 1 & 0 & 0\\
1 & 3 & 3 & 3 & 2 & 1 & 0\\
1 & 4 & 5 & 6 & 5 & 4 & 1
\end{array}
\right)  ^{-1}\left(
\begin{array}
[c]{c}%
0\\
0\\
1\\
1\\
1\\
1\\
1
\end{array}
\right)  =\left(
\begin{array}
[c]{c}%
0\\
0\\
1\\
0\\
-1\\
0\\
1
\end{array}
\right)  . \label{Q53}%
\end{equation}
Substituting the coefficients $Q\left(  3\right)  $ given by Eq. (\ref{Q53})
into Eq. (\ref{3333}) gives the canonical partition function\textbf{ }%
\begin{align}
Z_{3}\left(  \beta,5\right)   &  =0\times\left(  5\right)  \left(
e^{-\beta\varepsilon_{1}},e^{-\beta\varepsilon_{2}},...\right)  +0\times
\left(  4,1\right)  \left(  e^{-\beta\varepsilon_{1}},e^{-\beta\varepsilon
_{2}},...\right)  +1\times\left(  3,2\right)  \left(  e^{-\beta\varepsilon
_{1}},e^{-\beta\varepsilon_{2}},...\right) \nonumber\\
&  +0\times\left(  3,1^{2}\right)  \left(  e^{-\beta\varepsilon_{1}}%
,e^{-\beta\varepsilon_{2}},...\right)  -1\times\left(  2^{2},1\right)  \left(
e^{-\beta\varepsilon_{1}},e^{-\beta\varepsilon_{2}},...\right)  +0\times
\left(  2,1^{3}\right)  \left(  e^{-\beta\varepsilon_{1}},e^{-\beta
\varepsilon_{2}},...\right) \nonumber\\
&  +1\times\left(  1^{5}\right)  \left(  e^{-\beta\varepsilon_{1}}%
,e^{-\beta\varepsilon_{2}},...\right) \nonumber\\
&  =\left(  3,2\right)  \left(  e^{-\beta\varepsilon_{1}},e^{-\beta
\varepsilon_{2}},...\right)  -\left(  2^{2},1\right)  \left(  e^{-\beta
\varepsilon_{1}},e^{-\beta\varepsilon_{2}},...\right)  +\left(  2,1^{3}%
\right)  \left(  e^{-\beta\varepsilon_{1}},e^{-\beta\varepsilon_{2}%
},...\right) \nonumber\\
&  +\left(  1^{5}\right)  \left(  e^{-\beta\varepsilon_{1}},e^{-\beta
\varepsilon_{2}},...\right)  . \label{P53}%
\end{align}

For $q=2$, by Eq. (\ref{GAMMA}), we have
\begin{equation}
\Gamma\left(  2\right)  =\left(
\begin{array}
[c]{c}%
0\\
0\\
0\\
0\\
1\\
1\\
1
\end{array}
\right)  . \label{a14}%
\end{equation}
Substituting Eqs. (\ref{a10}) and (\ref{a14}) into Eq. (\ref{Qmatrix}) gives%
\begin{equation}
Q\left(  2\right)  =\left(
\begin{array}
[c]{ccccccc}%
1 & 0 & 0 & 0 & 0 & 0 & 0\\
1 & 1 & 0 & 0 & 0 & 0 & 0\\
1 & 1 & 1 & 0 & 0 & 0 & 0\\
1 & 2 & 1 & 1 & 0 & 0 & 0\\
1 & 2 & 2 & 1 & 1 & 0 & 0\\
1 & 3 & 3 & 3 & 2 & 1 & 0\\
1 & 4 & 5 & 6 & 5 & 4 & 1
\end{array}
\right)  ^{-1}\left(
\begin{array}
[c]{c}%
0\\
0\\
0\\
0\\
1\\
1\\
1
\end{array}
\right)  =\left(
\begin{array}
[c]{c}%
0\\
0\\
0\\
0\\
1\\
-1\\
0
\end{array}
\right)  . \label{Q52}%
\end{equation}
Substituting the coefficients $Q\left(  2\right)  $ given by Eq. (\ref{Q52})
into Eq. (\ref{3333}) gives the canonical partition function\textbf{ }%
\begin{align}
Z_{2}\left(  \beta,5\right)   &  =0\times\left(  5\right)  \left(
e^{-\beta\varepsilon_{1}},e^{-\beta\varepsilon_{2}},...\right)  +0\times
\left(  4,1\right)  \left(  e^{-\beta\varepsilon_{1}},e^{-\beta\varepsilon
_{2}},...\right)  +0\times\left(  3,2\right)  \left(  e^{-\beta\varepsilon
_{1}},e^{-\beta\varepsilon_{2}},...\right) \nonumber\\
&  +0\times\left(  3,1^{2}\right)  \left(  e^{-\beta\varepsilon_{1}}%
,e^{-\beta\varepsilon_{2}},...\right)  +1\times\left(  2^{2},1\right)  \left(
e^{-\beta\varepsilon_{1}},e^{-\beta\varepsilon_{2}},...\right)  -1\times
\left(  2,1^{3}\right)  \left(  e^{-\beta\varepsilon_{1}},e^{-\beta
\varepsilon_{2}},...\right) \nonumber\\
&  +0\times\left(  1^{5}\right)  \left(  e^{-\beta\varepsilon_{1}}%
,e^{-\beta\varepsilon_{2}},...\right) \nonumber\\
&  =\left(  2^{2},1\right)  \left(  e^{-\beta\varepsilon_{1}},e^{-\beta
\varepsilon_{2}},...\right)  -1\times\left(  2,1^{3}\right)  \left(
e^{-\beta\varepsilon_{1}},e^{-\beta\varepsilon_{2}},...\right)  . \label{P52}%
\end{align}

For $q=1$, by Eq. (\ref{GAMMA}), we have%
\begin{equation}
\Gamma\left(  1\right)  =\left(
\begin{array}
[c]{c}%
0\\
0\\
0\\
0\\
0\\
0\\
1
\end{array}
\right)  . \label{a15}%
\end{equation}
Substituting Eqs. (\ref{a10}) and (\ref{a15}) into Eq. (\ref{Qmatrix}) gives%

\begin{equation}
Q\left(  1\right)  =\left(
\begin{array}
[c]{ccccccc}%
1 & 0 & 0 & 0 & 0 & 0 & 0\\
1 & 1 & 0 & 0 & 0 & 0 & 0\\
1 & 1 & 1 & 0 & 0 & 0 & 0\\
1 & 2 & 1 & 1 & 0 & 0 & 0\\
1 & 2 & 2 & 1 & 1 & 0 & 0\\
1 & 3 & 3 & 3 & 2 & 1 & 0\\
1 & 4 & 5 & 6 & 5 & 4 & 1
\end{array}
\right)  ^{-1}\left(
\begin{array}
[c]{c}%
0\\
0\\
0\\
0\\
0\\
0\\
1
\end{array}
\right)  =\left(
\begin{array}
[c]{c}%
0\\
0\\
0\\
0\\
0\\
0\\
1
\end{array}
\right)  . \label{Q51}%
\end{equation}
Substituting the coefficients $Q\left(  2\right)  $ given by Eq. (\ref{Q52})
into Eq. (\ref{3333}) gives the canonical partition function\textbf{ }%
\begin{align}
Z_{1}\left(  \beta,5\right)   &  =0\times\left(  5\right)  \left(
e^{-\beta\varepsilon_{1}},e^{-\beta\varepsilon_{2}},...\right)  +0\times
\left(  4,1\right)  \left(  e^{-\beta\varepsilon_{1}},e^{-\beta\varepsilon
_{2}},...\right)  +0\times\left(  3,2\right)  \left(  e^{-\beta\varepsilon
_{1}},e^{-\beta\varepsilon_{2}},...\right) \nonumber\\
&  +0\times\left(  3,1^{2}\right)  \left(  e^{-\beta\varepsilon_{1}}%
,e^{-\beta\varepsilon_{2}},...\right)  +0\times\left(  2^{2},1\right)  \left(
e^{-\beta\varepsilon_{1}},e^{-\beta\varepsilon_{2}},...\right)  +0\times
\left(  2,1^{3}\right)  \left(  e^{-\beta\varepsilon_{1}},e^{-\beta
\varepsilon_{2}},...\right) \nonumber\\
&  +1\times\left(  1^{5}\right)  \left(  e^{-\beta\varepsilon_{1}}%
,e^{-\beta\varepsilon_{2}},...\right) \nonumber\\
&  =\left(  1^{5}\right)  \left(  e^{-\beta\varepsilon_{1}},e^{-\beta
\varepsilon_{2}},...\right)  . \label{P51}%
\end{align}

\subsection{The canonical partition function \label{cpf}}

In Appendix \ref{QIq}, we show how to calculate the coefficients $Q^{I}\left(
q\right)  $ from Eq. (\ref{Q}). The canonical partition function of ideal
Gentile gases can be represented as a linear combination of the $S$-functions
provided the coefficients $Q^{I}\left(  q\right)  $ are given.

In this section, we calculate the canonical partition function and express the
canonical partition function in terms of the single-particle partition
function defined by Eq. (\ref{single}), $Z\left(  \beta\right)  $.

$N=3$\textit{.} The simple characteristics $\chi_{I}^{K}$ of the permutation
group $S_{3}$ for $N=3$ is \cite{Littlewood2005The}
\begin{equation}
\chi=\left(
\begin{array}
[c]{ccc}%
1 & 1 & 1\\
2 & 0 & -1\\
1 & -1 & 1
\end{array}
\right)  , \label{c1}%
\end{equation}
where $\chi_{I}^{K}$ is expressed in the matrix form by taking the upper index
as the row index and the lower index as the column index. Substituting the
simple characteristics $\chi_{I}^{K}$, Eq. \textbf{(\ref{c1}),} into the
definition of the $S$-function, Eq. (\ref{ss}), gives the expression of
$S$-functions, $\left(  \lambda\right)  _{I}\left(  e^{-\beta\varepsilon_{1}%
},e^{-\beta\varepsilon_{2}},...\right)  $, in terms of the single-particle
partition functions $Z\left(  \beta\right)  $. Then substituting the
$S$-function $\left(  \lambda\right)  _{I}\left(  e^{-\beta\varepsilon_{1}%
},e^{-\beta\varepsilon_{2}},...\right)  $ into the canonical partition
function, Eq. \textbf{(}\ref{PARTITION2}\textbf{),} gives the expression of
the canonical partition function in terms of the single-particle partition
function $Z\left(  \beta\right)  $:
\begin{align}
Z_{2}\left(  \beta,3\right)   &  =\left(  2,1\right)  \left(  e^{-\beta
\varepsilon_{1}},e^{-\beta\varepsilon_{2}},...\right)  -\left(  1^{3}\right)
\left(  e^{-\beta\varepsilon_{1}},e^{-\beta\varepsilon_{2}},...\right)
\nonumber\\
&  =\frac{1}{3!}Z\left(  \beta\right)  ^{3}+\frac{1}{2}Z\left(  \beta\right)
Z\left(  2\beta\right)  -\frac{2}{3}Z\left(  3\beta\right)  . \label{c}%
\end{align}

$N=4$\textit{. }The simple characteristic $\chi_{I}^{K}$ of the permutation
group $S_{4}$ for $N=4$ is \cite{Littlewood2005The}%
\begin{equation}
\chi=\left(
\begin{array}
[c]{ccccc}%
1 & 1 & 1 & 1 & 1\\
3 & 0 & -1 & 1 & -1\\
2 & -1 & 2 & 0 & 0\\
3 & 0 & -1 & -1 & 1\\
1 & 1 & 1 & -1 & -1
\end{array}
\right)  . \label{c2}%
\end{equation}
Substituting the simple characteristics \textbf{(\ref{c2})} into the
$S$-function in Eqs. (\ref{P42}) and (\ref{P43}), respectively,\textit{ }gives%
\begin{align}
Z_{2}\left(  \beta,4\right)   &  =\left(  2^{2}\right)  \left(  e^{-\beta
\varepsilon_{1}},e^{-\beta\varepsilon_{2}},...\right)  -\left(  1^{4}\right)
\left(  e^{-\beta\varepsilon_{1}},e^{-\beta\varepsilon_{2}},...\right)
\nonumber\\
&  =\frac{1}{4!}Z\left(  \beta\right)  ^{4}+\frac{1}{4}Z\left(  \beta\right)
^{2}Z\left(  2\beta\right)  +\frac{1}{8}Z\left(  2\beta\right)  ^{2}-\frac
{2}{3}Z\left(  3\beta\right)  Z\left(  \beta\right)  +\frac{1}{4}Z\left(
4\beta\right)  \label{d}%
\end{align}
and%
\begin{align}
Z_{3}\left(  \beta,4\right)   &  =\left(  3,1\right)  \left(  e^{-\beta
\varepsilon_{1}},e^{-\beta\varepsilon_{2}},...\right)  -\left(  2,1^{2}%
\right)  \left(  e^{-\beta\varepsilon_{1}},e^{-\beta\varepsilon_{2}%
},...\right)  +\left(  1^{4}\right)  \left(  e^{-\beta\varepsilon_{1}%
},e^{-\beta\varepsilon_{2}},...\right) \nonumber\\
&  =\frac{1}{4!}Z\left(  \beta\right)  ^{4}+\frac{1}{4}Z\left(  \beta\right)
^{2}Z\left(  2\beta\right)  +\frac{1}{8}Z\left(  2\beta\right)  ^{2}+\frac
{1}{3}Z\left(  3\beta\right)  Z\left(  \beta\right)  -\frac{3}{4}Z\left(
4\beta\right)  . \label{e}%
\end{align}

$N=5$\textit{.} The simple characteristic $\chi_{I}^{K}$ of the permutation
group $S_{5}$ for $N=5$ is \cite{Littlewood2005The}%
\begin{equation}
\chi=\left(
\begin{array}
[c]{ccccccc}%
1 & 1 & 1 & 1 & 1 & 1 & 1\\
-1 & 0 & -1 & 1 & 0 & 1 & 4\\
0 & -1 & 1 & -1 & 1 & -1 & 5\\
1 & 0 & 0 & 0 & -2 & 0 & 6\\
0 & 1 & -1 & -1 & 1 & -1 & 5\\
-1 & 0 & 1 & 1 & 0 & 1 & 4\\
1 & -1 & -1 & 1 & 1 & 1 & 1
\end{array}
\right)  . \label{c3}%
\end{equation}
Substituting the simple characteristics \textbf{(}\ref{c3}\textbf{)} into the
$S$-functions in Eqs. (\ref{P52}), (\ref{P53}), and (\ref{P53}), respectively,
gives%
\begin{align}
Z_{2}\left(  \beta,5\right)   &  =\frac{1}{5!}Z\left(  \beta\right)
^{5}\nonumber\\
&  +\frac{1}{12}Z\left(  \beta\right)  ^{3}Z\left(  2\beta\right)  +\frac
{1}{8}Z\left(  \beta\right)  Z\left(  2\beta\right)  ^{2}\\
&  -\frac{1}{3}Z\left(  \beta\right)  ^{2}Z\left(  3\beta\right)  -\frac{1}%
{3}Z\left(  2\beta\right)  Z\left(  3\beta\right)  +\frac{1}{4}Z\left(
\beta\right)  Z\left(  4\beta\right)  +\frac{1}{5}Z\left(  5\beta\right)  ,
\label{f}%
\end{align}%
\begin{align}
Z_{3}\left(  \beta,5\right)   &  =\frac{1}{5!}Z\left(  \beta\right)
^{5}\nonumber\\
&  +\frac{1}{12}Z\left(  \beta\right)  ^{3}Z\left(  2\beta\right)  +\frac
{1}{8}Z\left(  \beta\right)  Z\left(  2\beta\right)  ^{2}\\
&  +\frac{1}{6}Z\left(  \beta\right)  ^{2}Z\left(  3\beta\right)  +\frac{1}%
{6}Z\left(  2\beta\right)  Z\left(  3\beta\right)  -\frac{3}{4}Z\left(
\beta\right)  Z\left(  4\beta\right)  +\frac{1}{5}Z\left(  5\beta\right)
\label{g}%
\end{align}
and%
\begin{align}
Z_{4}\left(  \beta,5\right)  =\frac{1}{5!}Z\left(  \beta\right)  ^{5}  &
+\frac{1}{12}Z\left(  \beta\right)  ^{3}Z\left(  2\beta\right)  +\frac{1}%
{8}Z\left(  \beta\right)  Z\left(  2\beta\right)  ^{2}\nonumber\\
&  +\frac{1}{6}Z\left(  \beta\right)  ^{2}Z\left(  3\beta\right)  +\frac{1}%
{6}Z\left(  2\beta\right)  Z\left(  3\beta\right)  +\frac{1}{4}Z\left(
\beta\right)  Z\left(  4\beta\right)  -\frac{4}{5}Z\left(  5\beta\right)  .
\label{h}%
\end{align}

$N=6$\textit{. }The simple characteristics $\chi_{I}^{K}$ of the permutation
group\textbf{ }$S_{6}$ for $N=6$ can be found in Ref.
\textbf{\cite{Littlewood2005The}}. A similar procedure gives
\begin{align}
Z_{2}\left(  \beta,6\right)   &  =\frac{1}{6!}Z\left(  \beta\right)
^{6}\nonumber\\
&  +\frac{1}{48}Z\left(  \beta\right)  ^{4}Z\left(  2\beta\right)  +\frac
{1}{16}Z\left(  \beta\right)  ^{2}Z\left(  2\beta\right)  ^{2}\nonumber\\
&  +\frac{1}{48}Z\left(  2\beta\right)  ^{3}-\frac{1}{9}Z\left(  \beta\right)
^{3}Z\left(  3\beta\right)  -\frac{1}{3}Z\left(  \beta\right)  Z\left(
2\beta\right)  Z\left(  3\beta\right) \nonumber\\
&  +\frac{2}{9}Z\left(  3\beta\right)  ^{2}+\frac{1}{8}Z\left(  \beta\right)
^{2}Z\left(  4\beta\right)  +\frac{1}{8}Z\left(  2\beta\right)  Z\left(
4\beta\right)  +\frac{1}{5}Z\left(  5\beta\right)  Z\left(  \beta\right)
-\frac{1}{3}Z\left(  6\beta\right)  , \label{i}%
\end{align}%
\begin{align}
Z_{3}\left(  \beta,6\right)   &  =\frac{1}{6!}Z\left(  \beta\right)
^{6}\nonumber\\
&  +\frac{1}{48}Z\left(  \beta\right)  ^{4}Z\left(  2\beta\right)  +\frac
{1}{16}Z\left(  \beta\right)  ^{2}Z\left(  2\beta\right)  ^{2}\nonumber\\
&  +\frac{1}{48}Z\left(  2\beta\right)  ^{3}+\frac{1}{18}Z\left(
\beta\right)  ^{3}Z\left(  3\beta\right)  +\frac{1}{6}Z\left(  \beta\right)
Z\left(  2\beta\right)  Z\left(  3\beta\right) \nonumber\\
&  +\frac{1}{18}Z\left(  3\beta\right)  ^{2}-\frac{3}{8}Z\left(  \beta\right)
^{2}Z\left(  4\beta\right)  -\frac{3}{8}Z\left(  2\beta\right)  Z\left(
4\beta\right)  +\frac{1}{5}Z\left(  5\beta\right)  Z\left(  \beta\right)
+\frac{1}{6}Z\left(  6\beta\right)  , \label{j}%
\end{align}%
\begin{align}
Z_{4}\left(  \beta,6\right)   &  =\frac{1}{6!}z\left(  \beta\right)
^{6}\nonumber\\
&  +\frac{1}{48}z\left(  \beta\right)  ^{4}z\left(  2\beta\right)  +\frac
{1}{16}z\left(  \beta\right)  ^{2}z\left(  2\beta\right)  ^{2}\nonumber\\
&  +\frac{1}{48}z\left(  2\beta\right)  ^{3}+\frac{1}{18}z\left(
\beta\right)  ^{3}z\left(  3\beta\right)  +\frac{1}{6}z\left(  \beta\right)
z\left(  2\beta\right)  z\left(  3\beta\right) \nonumber\\
&  +\frac{1}{18}z\left(  3\beta\right)  ^{2}+\frac{1}{8}z\left(  \beta\right)
^{2}z\left(  4\beta\right)  +\frac{1}{8}z\left(  2\beta\right)  z\left(
4\beta\right)  -\frac{4}{5}z\left(  5\beta\right)  z\left(  \beta\right)
+\frac{1}{6}z\left(  6\beta\right)  , \label{k}%
\end{align}
and%
\begin{align}
Z_{5}\left(  \beta,6\right)   &  =\frac{1}{6!}Z\left(  \beta\right)
^{6}\nonumber\\
&  +\frac{1}{48}Z\left(  \beta\right)  ^{4}Z\left(  2\beta\right)  +\frac
{1}{16}Z\left(  \beta\right)  ^{2}Z\left(  2\beta\right)  ^{2}\nonumber\\
&  +\frac{1}{48}Z\left(  2\beta\right)  ^{3}+\frac{1}{18}Z\left(
\beta\right)  ^{3}Z\left(  3\beta\right)  +\frac{1}{6}Z\left(  \beta\right)
Z\left(  2\beta\right)  Z\left(  3\beta\right) \nonumber\\
&  +\frac{1}{18}Z\left(  3\beta\right)  ^{2}+\frac{1}{8}Z\left(  \beta\right)
^{2}Z\left(  4\beta\right)  +\frac{1}{8}Z\left(  2\beta\right)  Z\left(
4\beta\right)  +\frac{1}{5}Z\left(  5\beta\right)  Z\left(  \beta\right)
-\frac{5}{6}Z\left(  6\beta\right)  . \label{l}%
\end{align}
The simple characteristics $\chi_{I}^{K}$ of the permutation group $S_{6}$ is
not listed in the paper, but it can be found in Ref. \cite{Littlewood2005The}.
The calculation of the canonical partition function for Gentile gases for
$N=6$ are not listed in this paper, but can be calculated by the method
provided in this paper.

From the results, Eqs. (\ref{c}), (\ref{d}), (\ref{e}), (\ref{f}), (\ref{g}),
(\ref{h}), (\ref{i}), (\ref{j}), (\ref{k}), and (\ref{l}), one can see that
the canonical partition functions $Z_{q}\left(  \beta,N\right)  $ of ideal
Gentile gases can be written in the form of $Z_{q}\left(  \beta,N\right)
=\frac{1}{N!}Z\left(  \beta\right)  ^{N}+$ corrections, where $Z\left(
\beta\right)  ^{N}$ is the canonical partition function of ideal classical
gases with $N$ particles and the correction depends on both $N$ and $q$.

\subsection{The virial coefficients}

In this section, we calculate the virial coefficients from the canonical
partition function.

The exact canonical partition function of ideal Gentile gases can be obtained
by substituting the single-particle partition function $Z\left(  \beta\right)
$, Eq. (\ref{single3}), into the canonical partition functions, Eqs.
(\ref{c}), (\ref{d}), (\ref{e}), (\ref{f}), (\ref{g}), (\ref{h}), (\ref{i}),
(\ref{j}), (\ref{k}), and (\ref{l}).

For $q=2$,%
\begin{equation}
Z_{2}\left(  \beta,3\right)  =\frac{1}{6}\left(  \frac{V}{\lambda^{3}}\right)
^{3}+\frac{1}{4\sqrt{2}}\left(  \frac{V}{\lambda^{3}}\right)  ^{2}-\frac
{2}{9\sqrt{3}}\frac{V}{\lambda^{3}}, \label{55}%
\end{equation}

\begin{equation}
Z_{2}\left(  \beta,4\right)  =\frac{1}{24}\left(  \frac{V}{\lambda^{3}%
}\right)  ^{4}+\frac{1}{8\sqrt{2}}\left(  \frac{V}{\lambda^{3}}\right)
^{3}+\frac{1}{64}\left(  \frac{V}{\lambda^{3}}\right)  ^{2}-\frac{2}{9\sqrt
{3}}\left(  \frac{V}{\lambda^{3}}\right)  ^{2}+\frac{1}{32}\left(  \frac
{V}{\lambda^{3}}\right)  ,
\end{equation}%
\begin{align}
Z_{2}\left(  \beta,5\right)   &  =\frac{1}{120}\left(  \frac{V}{\lambda^{3}%
}\right)  ^{5}+\frac{1}{24\sqrt{2}}\left(  \frac{V}{\lambda^{3}}\right)
^{4}+\frac{1}{64}\left(  \frac{V}{\lambda^{3}}\right)  ^{3}\nonumber\\
&  -\frac{1}{9\sqrt{3}}\left(  \frac{V}{\lambda^{3}}\right)  ^{3}+\frac{1}%
{32}\left(  \frac{V}{\lambda^{3}}\right)  ^{2}-\frac{1}{18\sqrt{6}}\left(
\frac{V}{\lambda^{3}}\right)  ^{2}+\frac{1}{25\sqrt{5}}\frac{V}{\lambda^{3}},
\end{align}%
\begin{align}
Z_{2}\left(  \beta,6\right)   &  =\frac{1}{720}\left(  \frac{V}{\lambda^{3}%
}\right)  ^{6}+\frac{1}{96\sqrt{2}}\left(  \frac{V}{\lambda^{3}}\right)
^{5}+\frac{1}{128}\left(  \frac{V}{\lambda^{3}}\right)  ^{4}-\frac{1}%
{27\sqrt{3}}\left(  \frac{V}{\lambda^{3}}\right)  ^{4}\nonumber\\
&  +\frac{1}{64}\left(  \frac{V}{\lambda^{3}}\right)  ^{3}-\frac{1}%
{768\sqrt{2}}\left(  \frac{V}{\lambda^{3}}\right)  ^{3}-\frac{1}{18\sqrt{6}%
}\left(  \frac{V}{\lambda^{3}}\right)  ^{3}+\frac{2}{243}\left(  \frac
{V}{\lambda^{3}}\right)  ^{2}\nonumber\\
&  +\frac{1}{128\sqrt{2}}\left(  \frac{V}{\lambda^{3}}\right)  ^{2}+\frac
{1}{25\sqrt{5}}\left(  \frac{V}{\lambda^{3}}\right)  ^{2}-\frac{1}{18\sqrt{6}%
}\left(  \frac{V}{\lambda^{3}}\right)  ^{2}. \label{66}%
\end{align}
Substituting Eqs. (\ref{55}-\ref{66}) into the equation of state (\ref{xtfc})
respectively gives the virial coefficients $a_{l}$ listed in Table
(\ref{table1}).

For $q=3$,%
\begin{equation}
Z_{3}\left(  \beta,3\right)  =\frac{1}{6}\left(  \frac{V}{\lambda^{3}}\right)
^{3}+\frac{1}{4\sqrt{2}}\left(  \frac{V}{\lambda^{3}}\right)  ^{2}+\frac
{1}{9\sqrt{3}}\frac{V}{\lambda^{3}}, \label{77}%
\end{equation}

\begin{align}
Z_{3}\left(  \beta,4\right)   &  =\frac{1}{24}\left(  \frac{V}{\lambda^{3}%
}\right)  ^{4}+\frac{1}{8\sqrt{2}}\left(  \frac{V}{\lambda^{3}}\right)
^{3}+\frac{1}{64}\left(  \frac{V}{\lambda^{3}}\right)  ^{2}\nonumber\\
&  +\frac{1}{9\sqrt{3}}\left(  \frac{V}{\lambda^{3}}\right)  ^{2}-\frac{3}%
{32}\left(  \frac{V}{\lambda^{3}}\right)  ,
\end{align}%
\begin{align}
Z_{3}\left(  \beta,5\right)   &  =\frac{1}{120}\left(  \frac{V}{\lambda^{3}%
}\right)  ^{5}+\frac{1}{24\sqrt{2}}\left(  \frac{V}{\lambda^{3}}\right)
^{4}+\frac{1}{64}\left(  \frac{V}{\lambda^{3}}\right)  ^{3}\nonumber\\
&  +\frac{1}{18\sqrt{3}}\left(  \frac{V}{\lambda^{3}}\right)  ^{3}-\frac
{3}{32}\left(  \frac{V}{\lambda^{3}}\right)  ^{2}+\frac{1}{36\sqrt{6}}\left(
\frac{V}{\lambda^{3}}\right)  ^{2}+\frac{1}{25\sqrt{5}}\frac{V}{\lambda^{3}},
\end{align}%
\begin{align}
Z_{3}\left(  \beta,6\right)   &  =\frac{1}{720}\left(  \frac{V}{\lambda^{3}%
}\right)  ^{6}+\frac{1}{96\sqrt{2}}\left(  \frac{V}{\lambda^{3}}\right)
^{5}+\frac{1}{128}\left(  \frac{V}{\lambda^{3}}\right)  ^{4}+\frac{1}%
{54\sqrt{3}}\left(  \frac{V}{\lambda^{3}}\right)  ^{4}\nonumber\\
&  -\frac{3}{64}\left(  \frac{V}{\lambda^{3}}\right)  ^{3}+\frac{1}%
{768\sqrt{2}}\left(  \frac{V}{\lambda^{3}}\right)  ^{3}+\frac{1}{36\sqrt{6}%
}\left(  \frac{V}{\lambda^{3}}\right)  ^{3}+\frac{1}{486}\left(  \frac
{V}{\lambda^{3}}\right)  ^{2}\nonumber\\
&  -\frac{3}{128\sqrt{2}}\left(  \frac{V}{\lambda^{3}}\right)  ^{2}+\frac
{1}{25\sqrt{5}}\left(  \frac{V}{\lambda^{3}}\right)  ^{2}+\frac{1}{36\sqrt{6}%
}\left(  \frac{V}{\lambda^{3}}\right)  ^{2}. \label{88}%
\end{align}
Substituting Eqs. (\ref{77})-(\ref{88}) into the equation of state
(\ref{xtfc}) respectively gives the virial coefficients $a_{l}$ listed in
Table (\ref{table2}).

For $q=4$,
\begin{align}
Z_{4}\left(  \beta,4\right)   &  =\frac{1}{24}\left(  \frac{V}{\lambda^{3}%
}\right)  ^{4}+\frac{1}{8\sqrt{2}}\left(  \frac{V}{\lambda^{3}}\right)
^{3}+\frac{1}{64}\left(  \frac{V}{\lambda^{3}}\right)  ^{2}\nonumber\\
&  +\frac{1}{9\sqrt{3}}\left(  \frac{V}{\lambda^{3}}\right)  ^{2}+\frac{1}%
{32}\left(  \frac{V}{\lambda^{3}}\right)  , \label{99}%
\end{align}%
\begin{align}
Z_{4}\left(  \beta,5\right)   &  =\frac{1}{120}\left(  \frac{V}{\lambda^{3}%
}\right)  ^{5}+\frac{1}{24\sqrt{2}}\left(  \frac{V}{\lambda^{3}}\right)
^{4}+\frac{1}{64}\left(  \frac{V}{\lambda^{3}}\right)  ^{3}\nonumber\\
&  +\frac{1}{18\sqrt{3}}\left(  \frac{V}{\lambda^{3}}\right)  ^{3}+\frac
{1}{32}\left(  \frac{V}{\lambda^{3}}\right)  ^{2}+\frac{1}{36\sqrt{6}}\left(
\frac{V}{\lambda^{3}}\right)  ^{2}-\frac{4}{25\sqrt{5}}\frac{V}{\lambda^{3}},
\end{align}%
\begin{align}
Z_{4}\left(  \beta,6\right)   &  =\frac{1}{720}\left(  \frac{V}{\lambda^{3}%
}\right)  ^{6}+\frac{1}{96\sqrt{2}}\left(  \frac{V}{\lambda^{3}}\right)
^{5}+\frac{1}{128}\left(  \frac{V}{\lambda^{3}}\right)  ^{4}+\frac{1}%
{54\sqrt{3}}\left(  \frac{V}{\lambda^{3}}\right)  ^{4}\nonumber\\
&  +\frac{1}{64}\left(  \frac{V}{\lambda^{3}}\right)  ^{3}+\frac{1}%
{768\sqrt{2}}\left(  \frac{V}{\lambda^{3}}\right)  ^{3}+\frac{1}{36\sqrt{6}%
}\left(  \frac{V}{\lambda^{3}}\right)  ^{3}+\frac{1}{486}\left(  \frac
{V}{\lambda^{3}}\right)  ^{2}\nonumber\\
&  +\frac{1}{128\sqrt{2}}\left(  \frac{V}{\lambda^{3}}\right)  ^{2}-\frac
{4}{25\sqrt{5}}\left(  \frac{V}{\lambda^{3}}\right)  ^{2}+\frac{1}{36\sqrt{6}%
}\left(  \frac{V}{\lambda^{3}}\right)  ^{2}. \label{10}%
\end{align}
Substituting Eqs. (\ref{99}-\ref{10}) into the equation of state (\ref{xtfc})
respectively gives the virial coefficients $a_{l}$ listed in Table
(\ref{table3}).

For $q=5$,%
\begin{align}
Z_{5}\left(  \beta,5\right)   &  =\frac{1}{120}\left(  \frac{V}{\lambda^{3}%
}\right)  ^{5}+\frac{1}{24\sqrt{2}}\left(  \frac{V}{\lambda^{3}}\right)
^{4}+\frac{1}{64}\left(  \frac{V}{\lambda^{3}}\right)  ^{3}\nonumber\\
&  +\frac{1}{18\sqrt{3}}\left(  \frac{V}{\lambda^{3}}\right)  ^{3}+\frac
{1}{32}\left(  \frac{V}{\lambda^{3}}\right)  ^{2}+\frac{1}{36\sqrt{6}}\left(
\frac{V}{\lambda^{3}}\right)  ^{2}+\frac{1}{25\sqrt{5}}\frac{V}{\lambda^{3}},
\label{d1}%
\end{align}%
\begin{align}
Z_{5}\left(  \beta,6\right)   &  =\frac{1}{720}\left(  \frac{V}{\lambda^{3}%
}\right)  ^{6}+\frac{1}{96\sqrt{2}}\left(  \frac{V}{\lambda^{3}}\right)
^{5}+\frac{1}{128}\left(  \frac{V}{\lambda^{3}}\right)  ^{4}+\frac{1}%
{54\sqrt{3}}\left(  \frac{V}{\lambda^{3}}\right)  ^{4}\nonumber\\
&  +\frac{1}{64}\left(  \frac{V}{\lambda^{3}}\right)  ^{3}+\frac{1}%
{768\sqrt{2}}\left(  \frac{V}{\lambda^{3}}\right)  ^{3}+\frac{1}{36\sqrt{6}%
}\left(  \frac{V}{\lambda^{3}}\right)  ^{3}+\frac{1}{486}\left(  \frac
{V}{\lambda^{3}}\right)  ^{2}\nonumber\\
&  +\frac{1}{128\sqrt{2}}\left(  \frac{V}{\lambda^{3}}\right)  ^{2}+\frac
{1}{25\sqrt{5}}\left(  \frac{V}{\lambda^{3}}\right)  ^{2}-\frac{5}{36\sqrt{6}%
}\left(  \frac{V}{\lambda^{3}}\right)  ^{2}. \label{d2}%
\end{align}
Substituting Eqs. (\ref{d1}) and (\ref{d2}) into Eq. (\ref{xtfc}) respectively
gives the virial coefficients $a_{l}$ listed in Table (\ref{table4}).

For $q=6$,%
\begin{align}
Z_{6}\left(  \beta,6\right)   &  =\frac{1}{720}\left(  \frac{V}{\lambda^{3}%
}\right)  ^{6}+\frac{1}{96\sqrt{2}}\left(  \frac{V}{\lambda^{3}}\right)
^{5}+\frac{1}{128}\left(  \frac{V}{\lambda^{3}}\right)  ^{4}+\frac{1}%
{54\sqrt{3}}\left(  \frac{V}{\lambda^{3}}\right)  ^{4}\nonumber\\
&  +\frac{1}{64}\left(  \frac{V}{\lambda^{3}}\right)  ^{3}+\frac{1}%
{768\sqrt{2}}\left(  \frac{V}{\lambda^{3}}\right)  ^{3}+\frac{1}{36\sqrt{6}%
}\left(  \frac{V}{\lambda^{3}}\right)  ^{3}+\frac{1}{486}\left(  \frac
{V}{\lambda^{3}}\right)  ^{2}\nonumber\\
&  +\frac{1}{128\sqrt{2}}\left(  \frac{V}{\lambda^{3}}\right)  ^{2}+\frac
{1}{25\sqrt{5}}\left(  \frac{V}{\lambda^{3}}\right)  ^{2}+\frac{1}{36\sqrt{6}%
}\left(  \frac{V}{\lambda^{3}}\right)  ^{2}. \label{d3}%
\end{align}
Substituting Eq (\ref{d3}) into the equation of state (\ref{xtfc}) gives the
virial coefficients $a_{l}$ listed in Table (\ref{table5}).

\acknowledgments

We are very indebted to Dr G. Zeitrauman for his encouragement. This work is supported in part by NSF of China under Grant
No. 11575125 and No. 11675119.










\providecommand{\href}[2]{#2}\begingroup\raggedright\endgroup


\end{document}